\documentclass[11pt,a4paper]{article}
\usepackage[latin1]{inputenc}
\usepackage{pdflscape}

\usepackage{graphicx,booktabs}
\usepackage{multicol,wrapfig}

\usepackage{amsfonts}
\usepackage{amsmath}
\usepackage{amssymb}
\usepackage{theorem}
\usepackage{upgreek}
\usepackage{xcolor}
\usepackage{float}
\usepackage{rotating}

\theoremstyle{change}
\theorembodyfont{\slshape}
\newtheorem{satz}{Theorem}[section]

\newtheorem{prop}[satz]{Proposition}
\newtheorem{kor}[satz]{Corollary}

\theorembodyfont{\upshape\small}
\newtheorem{bsp}[satz]{Example}
\newtheorem{bem}[satz]{Remark}

\theorembodyfont{\ttfamily}

\newenvironment{proof}{\list{}{\itemindent-\leftmargin}%
                \item\textbf{Proof: }\small}{\hbox{}\hfill\#\newline\endlist\vspace{-2mm}}

\newcommand{\ba}{\begin{equation}}
\newcommand{\ea}{\end{equation}}

\usepackage{dsfont}
\newcommand{\indfkt}{\mathds{1}}

\newcommand{\0}{\mbox{\boldmath $0$}}

\newcommand{\fp}{\mbox{\boldmath $p$}}

\newcommand{\fx}{\mbox{\boldmath $x$}}

\newcommand{\fZ}{\mbox{\boldmath $Z$}}

\newcommand{\mP}{\mbox{\textup{\textbf{P}}}}
\newcommand{\mT}{\mbox{\textup{\textbf{T}}}}

\newcommand{\fSigma}{\mbox{\boldmath $\Sigma$}}

\newcommand{\norm}{\textup{N}}

\newcommand{\bbn}{\mathbb{N}}

\newcommand{\cat}{m}

\newcommand{\iid}{i.\,i.\,d.}
\newcommand{\ie}{i.\,e., }
\newcommand{\eg}{e.\,g., }

\usepackage{natbib}

\usepackage{hyperref}

\begin{document}



\parindent 0cm

\title{Transcripts and Algebraic Distances in Time Series: Stochastic Properties and Nonparametric Dependence Tests}

\author{
Christian H.\ Wei\ss{}\thanks{Department of Mathematics and Statistics, Helmut Schmidt University, 22043 Hamburg, Germany}\ \thanks{Corresponding author. E-Mail: \href{mailto:weissc@hsu-hh.de}{\nolinkurl{weissc@hsu-hh.de}}. ORCID: \href{https://orcid.org/0000-0001-8739-6631}{\nolinkurl{0000-0001-8739-6631}}.}
\and 
Jos\'e M.\ Amig\'o\thanks{Centro de Investigaci\'on Operativa, Universidad Miguel Hern\'andez, 03202 Elche, Spain}\ \thanks{E-Mail: \href{mailto:jm.amigo@umh.es}{\nolinkurl{jm.amigo@umh.es}}. ORCID: \href{https://orcid.org/0000-0002-1642-1171}{\nolinkurl{0000-0002-1642-1171}}.}
}

\maketitle

\begin{abstract}
The use of ordinal patterns (OPs) for analyzing the dependence structure of univariate and continuously distributed processes has gained popularity in recent years. This research goes one step further and considers the transcripts being computed from successive OPs in the time series. Transcripts constitute a kind of ``difference'' between successive OPs and thus naturally relate to two algebraic distances between OPs, the Cayley and Kendall edit distances. The original time series is transformed into a sequence
of transcripts or distances, respectively, and important stochastic properties thereof are derived. It is shown that these properties differ substantially among different types of original processes. This motivates the development of various statistics based on transcripts and edit distances in order to investigate the dependence structure of the original process. In particular, the asymptotic distribution of these statistics under the null hypothesis of serial independence is derived, which is then used to implement nonparametric tests for serial dependence. A simulation study shows that these novel dependence tests have appealing power properties, often outperforming former OP-based dependence tests. A concluding real-world data example illustrates the application and interpretation of the proposed approaches in practice.

\end{abstract}

\section{Introduction}
\label{Introduction}

Among the symbolic representations of a univariate
real-valued time series, the symbolization scheme that ranks the entries in a sliding window according to magnitude has been very well received in the data analysis community since its inception in 2002 \citep{bandt02}. The rank vectors obtained in this way are called ordinal patterns (OPs), their length being defined as the size of the window used for symbolization, and the resulting symbolic time series is called an ordinal representation of the original time series. Reasons for the popularity of ordinal representations include conceptual simplicity, fast computation, relative robustness to additive noise, and the possibility of computation in real time (since it is not necessary to know the range of values). More important for our purposes, OPs of a given length can be viewed as permutations of a set whose cardinality is the length of the OPs. Therefore, unlike what happens with other symbolization or discretization schemes, the symbols of an ordinal representation belong to a group, namely the set of all such permutations
endowed with function composition.

\smallskip
The algebraic structure of the OPs provides additional leverage for the analysis of time series in ordinal representations. To the best of our knowledge, the first proposal to harness the algebraic nature of OPs, made by \citet{monetti09}, was the concept of transcript between two group elements. As we will see below, transcripts can be interpreted as a generalization of the concept of difference (or subtraction) in additive groups (such as the real numbers). Recently, a relationship between transcripts and certain edit distances has drawn the attention of data analysts \citep{amigo25}. Indeed, the Cayley and Kendall edit distances between two permutations (viewed as symbolic words) can be formulated as norms of the transcript of those permutations. This connection makes transcripts even more interesting in time series analysis.

\smallskip
Traditional transcript-based tools in time series analysis include entropies, mutual information, divergences, statistical complexities, coupling complexity coefficients, and more \citep{monetti09,amigo12,monetti13}. Typical applications include synchronization detection and information directionality in coupled dynamics, classification and discrimination of random and deterministic processes, and characterization of interactions \citep{amigo23}. These well-known tools and applications have to be complemented by the new arrivals: the Cayley and Kendall edit distances. These two metrics have emerged as promising tools in practical applications due to their satisfactory results in generalized synchronization detection \citep{amigo25}.

\bigskip
In this paper, we study the performance of transcripts, Cayley distance and Kendall distance in the detection of serial dependence in univariate time series. More specifically, we derive their stochastic properties and, consequently, propose transcript-based (including distance-based) statistics for nonparametric serial dependence tests. These tests are carried out using
both numerical simulations and real-world data. As it turns out, the novel statistics generally outperform previous OP-based statistics \citep{weiss22}, which confirms the transcripts, Cayley and Kendall distances as useful tools for data analysts.

\smallskip
This paper is structured as follows. In Section~\ref{Preliminaries}, we introduce the conceptual background: ordinal patterns, permutations, symmetric groups, transcripts, and the Cayley and Kendall distances. Specifically, Section~\ref{The symmetric group of degree 3} focuses on the symmetric group of degree 3, because OPs of length 3 are particularly relevant for applications and are, thus, used in the rest of this paper. The study of the stochastic properties of transcripts and, hence, of distances begins in Section~\ref{Distributional Properties of Transcripts}, first with the
marginal distributions (Section~\ref{Marginal Distribution of Transcripts}), followed by the bivariate
distributions (Section~\ref{Bivariate Distributions of Transcripts}). With serial-dependence testing in mind, Section~\ref{Asymptotics of Transcript Statistics} is devoted to the asymptotics of transcript statistics under the assumption that the original process is independent and identically distributed (\iid). In addition to the mean Cayley and Kendall distances, we use another two transcript-based statistics: the deviance (or Kullback--Leibler divergence) and the Pearson statistic for goodness-of-fit. The empirical sizes and power values of the corresponding tests are measured using numerical simulations in Section~\ref{Simulation-based Performance Analyses} and real-world data in Section~\ref{Real-World Data Application}. A summary, the main conclusions, and an outlook are the contents of the final Section~\ref{Conclusions}.

\numberwithin{satz}{subsection}


\section{Preliminaries}
\label{Preliminaries}


\subsection{Permutations, transcripts, and distances}
\label{Permutations, transcripts, and distances}

Let $(X_{t})=(X_{t})_{t\in \mathbb{Z}=\{\ldots ,-1,0,1,\ldots \}}$ be a 
real-valued and continuously distributed process, and let $m\in \mathbb{N}=\{1,2,\ldots \}$ with $m\geq 2$
be the length of the considered \emph{ordinal patterns} (``$m$-OPs'', or just ``OPs'' if~$m$ is clear), see %
\citet{bandt02}. Let $S_{m}$ denote the symmetric group of degree~$m$, which
consists of the $m!$ possible permutations of the integers $\{1,\ldots ,m\}$
endowed with function composition; the group $S_{m}$ is non-abelian for $m\geq 3$. The permutations~$\pi \in S_{m}$ can be
used in various ways to represent the $m!$ different OPs of a vector $%
\mbox{\boldmath $x$}=(x_{1},\ldots ,x_{m})\in \mathbb{R}^{m}$, see %
\citet{berger19}. Here, we consider the \emph{permutation representation},
\ie we use $\pi =(i_{1},\ldots ,i_{m})\in S_{m}$ for expressing
that permutation, which causes an ascending order of the components of~$\fx$: 
\begin{equation}
x_{i_{1}}\leq x_{i_{2}}\leq \ldots \leq x_{i_{m}},\quad \text{and}\quad
i_{k-1}<i_{k}\text{ if }x_{i_{k-1}}=x_{i_{k}}\text{ for }k\geq 2.
\label{ordpattern}
\end{equation}
The second case in \eqref{ordpattern} refers to the possible occurrence of ties in~$\fx$. Since we assume the data-generating process $(X_t)$ to be continuously distributed, ties happen with probability zero. However, in real-world applications with limited measurement accuracy, ties may happen anyway, but we assume them to occur at an negligible rate.

\smallskip
By mapping the $t$th segment $\mbox{\boldmath $X$}_{t}=(X_{t},\ldots
,X_{t+m-1})$ of the process $(X_{t})$ onto its corresponding OP~$\pi _{t}$,
the originally real-valued process $(X_{t})$ is discretized and transformed
into the OP-series $(\pi _{t})=(\pi _{t})_{t\in \mathbb{Z}}$. 
The OP-series $(\pi _{t})$ is called the \textit{ordinal representation} of length~$m$ of
the time series $(x_{t})=(x_{t})_{t\in \mathbb{Z}}$, a realization of $%
(X_{t})$.

\medskip
Choosing for the permutation $\pi :n\mapsto i_{n}$ ($1\leq n,i_{n}\leq m$) the one-line notation $(i_{1},i_{2},\ldots ,i_{m})$ from among other possibilities (\eg matrix forms) allows $\pi$ to be viewed as
the word $i_{1}i_{2}\ldots i_{m}$ composed by letters from the alphabet $%
\{1,2,\ldots,m\}$. In turn, this will allow us below to leverage edit distances
between permutations for the statistical analysis of OP-series $(\pi _{t})$.

\medskip
As said above, the symmetric group~$S_{m}$ forms an algebraic group together
with the composition \textquotedblleft $\circ $\textquotedblright\ of
permutations. Following the usual convention in the literature, here the composition $\pi
_{1}\circ \pi _{2}$ is meant as a \textquotedblleft right
action\textquotedblright , \ie the permutation $\pi _{1}$ acts first and $%
\pi _{2}$ acts second, so that $(\pi _{1}\circ \pi _{2})(n)=\pi _{2}\big(\pi
_{1}(n)\big)$ (usually written as $(n)(\pi _{1}\circ \pi _{2})=\big((n)\pi _{1}\big)\pi_{2}$) for $n=1,\ldots,m$. Therefore, if $\pi _{1}=(i_{1},i_{2},\ldots,i_{m})$
and $\pi _{2}=(j_{1},j_{2},\ldots,j_{m})$, then%
\begin{equation}
\pi _{1}\circ \pi _{2}=(j_{i_{1}},j_{i_{2}},\ldots,j_{i_{m}}).  \label{pi1 x p2}
\end{equation}
Precisely, to exploit the algebraic structure of $S_{m}$ in
the analysis of ordinal representations, \citet{monetti09} introduced the
concept of \emph{transcripts} between OPs. Formally, the transcript $\tau
:S_{m}\times S_{m}\rightarrow S_{m}$ maps two ordinal patterns $\pi _{1},\pi
_{2}\in S_{m}$ on the OP~$\pi \in S_{m}$ defined by%
\begin{equation}
\pi =\tau (\pi _{1},\pi _{2}):=\pi _{2}\circ \pi _{1}^{-1},
\label{transcrip}
\end{equation}
called the transcript between $\pi_{1}$ and $\pi_{2}$. Hence, $\pi =\tau (\pi _{1},\pi _{2})$ is equal to that permutation which
transforms~$\pi _{1}$ into~$\pi _{2}$, in the sense that $\pi \circ \pi_{1}=\pi _{2}$. For this reason, $\tau (\pi _{1},\pi _{2})$ is sometimes called the transcript from the source $\pi_{1}$ to the source $\pi_{2}$ \citep{monetti09}.
Considering the analogy to transcripts in an additive group, where $\tau (x,y)=y-x$, we may interpret the transcript \eqref{transcrip} as a kind of ``difference'' or ``dissimilarity'' between the OPs~$\pi _{1}$ and~$\pi _{2}$. 
As such, transcripts are related to the following two types
of algebraic distance between two permutations \citep{amigo25}:

\begin{itemize}
\item The \emph{Cayley distance} $d_{\textup{C}}:S_{m}\times
S_{m}\rightarrow \mathbb{N}_{0}=\{0,1,\ldots \}$ between $\pi _{1},\pi
_{2}\in S_{m}$ is defined as the minimum number of transpositions needed to
transform~$\pi _{1}$ into~$\pi _{2}$.

\item The \emph{Kendall distance} $d_{\textup{K}}:S_{m}\times
S_{m}\rightarrow \mathbb{N}_{0}=\{0,1,\ldots \}$ between $%
\pi _{1},\pi _{2}\in S_{m}$ is defined as the minimum number of adjacent
transpositions needed to transform~$\pi _{1}$ into~$\pi _{2}$.
\end{itemize}

By definition,
\begin{equation}
d_{\textup{C}}(\pi_{1},\pi_{2})\leq d_{\textup{K}}(\pi_{1},\pi_{2})
\label{d_C <= d_K}
\end{equation}%
for all $\pi_{1},\pi_{2}\in S_{m}$. A quantitative relationship between the Cayley and Kendall distances and transcripts is given in the following proposition (see \citealp{nguyen24,kendall1938}).

\begin{prop}
\label{PropComput dC dK} (a) Let $\pi%
=(i_{1},\ldots,i_{m})\in S_{m}$ and $C(\pi)$ be the number of
cycles (including $1$-cycles) in the cycle factorization of the permutation $%
\pi$. Then,%
\begin{equation}
d_{\textup{C}}(\pi_{1},\pi_{2})=m-C\big(\tau(\pi_{1},\pi_{2})\big)  \label{Comput d_C}
\end{equation}%
for all $\pi_{1},\pi_{2} \in S_{m}$.

\smallskip
(b) Let $I(\pi)$ be the number of inversions in the permutation $\pi$, \ie the number of ordered pairs $(i_{r},i_{s})$, $1\leq
r<s\leq m$, such that $i_{r}>i_{s}$. Then,
\begin{equation}
d_{\textup{K}}(\pi_{1},\pi_{2})=I\big(\tau(\pi_{1},\pi_{2})\big)  \label{Comput d_K}
\end{equation}%
for all $\pi_{1},\pi_{2} \in S_{m}$.
\end{prop}

The upper bounds of $d_{\textup{C}}$ and $d_{\textup{K}}$ in $S_{m}$ follow from equations \eqref{Comput d_C} and \eqref{Comput d_K}, namely,
\begin{equation}
d_{\textup{C}}(\pi_{1},\pi_{2})\leq m-1 \;\;\mbox{and} \;\;
d_{\textup{K}}(\pi_{1},\pi_{2})\leq \frac{m(m-1)}{2}.
\label{d_CK max}
\end{equation}%

Cayley and Kendall distances are examples of edit distances, \ie distances (in axiomatic sense) between two string of symbols, defined as the minimum number of allowed edit operations (insertions, deletions, substitutions, transpositions) that transform one string into the other. Edit distances between permutations were proposed by \citet{sorensen07}.

\medskip
In time series analysis---the framework of this paper---transcripts can 
refer (i) to pairs of elements from a single time series $(X_{t})$ (one element shifted in time with respect to the other) to detect, \eg serial dependence or, else, (ii) to pairs of (simultaneous or time-shifted) elements from two coupled time series $(X_{t})$ and $(Y_{t})$ to study, \eg different synchronization regimes. If ($\pi_{t}$) (resp.\ $(\rho_{t})$) is the ordinal representation of length~$m$ of $(X_{t})$ (resp.\ $(Y_{t})$), then the transcripts   
\begin{equation}
\tau_{t}(\pi_{t},\pi_{t+\Lambda})=(\pi_{t+\Lambda}\circ\pi_{t}^{-1}) \;\; 
(\mbox{resp.} \; \tau_{t}(\pi_{t},\rho_{t+\Lambda})=\rho_{t+\Lambda}\circ\pi_{t}^{-1}) 
\label{self and cross}
\end{equation}%
are called \textit{self-transcripts} of $X_{t}$ (resp.\ \textit{cross-transcripts} of $X_{t}$ and $Y_{t}$) with \textit{coupling delay} $\Lambda\in \mathbb{Z}$ \citep{amigo12}. In this paper, our focus is on the self-transcripts obtained from a univariate time series (henceforth simply called transcripts), whereas an analysis of the cross-transcripts in multivariate time series is recommended for future research.

\subsection{The symmetric group of degree 3}
\label{The symmetric group of degree 3}

In what follows, we solely focus on OPs of length $m=3$ (so ``3-OPs'', to be more precise), which is the most
common choice in practice, see \citet{bandt07,bandt19}. Then, the six
possible OPs (in lexicographic order) are the following:

\ba
\label{OrdPatt3}
\begin{array}{@{}lcccccc@{}}
\text{Perm.:} & (1,2,3) & (1,3,2) & (2,1,3) & (2,3,1) & (3,1,2) & (3,2,1) \\
\text{Obs.:}\\[-3ex]
 & \includegraphics[viewport=75 85 135 145, clip=, scale=0.65]{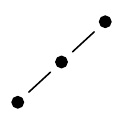}
 & \includegraphics[viewport=75 85 135 145, clip=, scale=0.65]{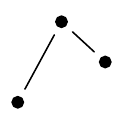} 
 & \includegraphics[viewport=75 85 135 145, clip=, scale=0.65]{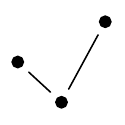} 
 & \includegraphics[viewport=75 85 135 145, clip=, scale=0.65]{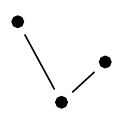} 
 & \includegraphics[viewport=75 85 135 145, clip=, scale=0.65]{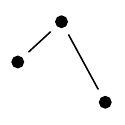} 
 & \includegraphics[viewport=75 85 135 145, clip=, scale=0.65]{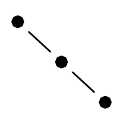} 
\end{array}
\ea
For the case $m=3$ considered in this research, the Cayley table of permutations is given by (see equation \eqref{pi1 x p2}) 

\ba
\label{tabOP}
\begin{array}{c|cccccc}
\toprule
\text{Composition } \circ & \pi^{[1]} & \pi^{[2]} & \pi^{[3]} & \pi^{[4]} & \pi^{[5]} & \pi^{[6]} \\
\midrule
\pi^{[1]} = (1,2,3) & \pi^{[1]} & \pi^{[2]} & \pi^{[3]} & \pi^{[4]} & \pi^{[5]} & \pi^{[6]} \\
\pi^{[2]} = (1,3,2) & \pi^{[2]} & \pi^{[1]} & \pi^{[4]} & \pi^{[3]} & \pi^{[6]} & \pi^{[5]} \\
\pi^{[3]} = (2,1,3) & \pi^{[3]} & \pi^{[5]} & \pi^{[1]} & \pi^{[6]} & \pi^{[2]} & \pi^{[4]} \\
\pi^{[4]} = (2,3,1) & \pi^{[4]} & \pi^{[6]} & \pi^{[2]} & \pi^{[5]} & \pi^{[1]} & \pi^{[3]} \\
\pi^{[5]} = (3,1,2) & \pi^{[5]} & \pi^{[3]} & \pi^{[6]} & \pi^{[1]} & \pi^{[4]} & \pi^{[2]} \\
\pi^{[6]} = (3,2,1) & \pi^{[6]} & \pi^{[4]} & \pi^{[5]} & \pi^{[2]} & \pi^{[3]} & \pi^{[1]} \\
\bottomrule
\end{array}
\ea

While $(\pi^{[1]})^{-1}=\pi^{[1]}=\mathrm{id}$ trivially holds as it constitutes the identity element of~$S_3$, note that
\begin{equation}
(\pi^{[2]})^{-1}=\pi^{[2]}, \; (\pi^{[3]})^{-1}=\pi^{[3]}, \;
(\pi^{[6]})^{-1}=\pi^{[6]},
\label{order2}
\end{equation}%
whereas
\begin{equation}
(\pi^{[4]})^{-1}=\pi^{[5]}, \; (\pi^{[5]})^{-1}=\pi^{[4]}.
\label{order3}
\end{equation}%

Hence, the Cayley table of transcripts is equal to (see equation \eqref{transcrip}) 

\ba
\label{tab_tr}
\begin{array}{c|cccccc}
\toprule
\text{Transcript } \tau & \pi^{[1]} & \pi^{[2]} & \pi^{[3]} & \pi^{[4]} & \pi^{[5]} & \pi^{[6]} \\
\midrule
\pi^{[1]} = (1,2,3) & \pi^{[1]} & \pi^{[2]} & \pi^{[3]} & \pi^{[4]} & \pi^{[5]} & \pi^{[6]} \\
\pi^{[2]} = (1,3,2) & \pi^{[2]} & \pi^{[1]} & \pi^{[5]} & \pi^{[6]} & \pi^{[3]} & \pi^{[4]} \\
\pi^{[3]} = (2,1,3) & \pi^{[3]} & \pi^{[4]} & \pi^{[1]} & \pi^{[2]} & \pi^{[6]} & \pi^{[5]} \\
\pi^{[4]} = (2,3,1) & \pi^{[5]} & \pi^{[6]} & \pi^{[2]} & \pi^{[1]} & \pi^{[4]} & \pi^{[3]} \\
\pi^{[5]} = (3,1,2) & \pi^{[4]} & \pi^{[3]} & \pi^{[6]} & \pi^{[5]} & \pi^{[1]} & \pi^{[2]} \\
\pi^{[6]} = (3,2,1) & \pi^{[6]} & \pi^{[5]} & \pi^{[4]} & \pi^{[3]} & \pi^{[2]} & \pi^{[1]} \\
\bottomrule
\end{array}
\ea

From this table, we can derive the sets~$\Pi _{k}$ of OP-pairs corresponding
to the transcript~$\pi ^{[k]}$, $k=1,\ldots,m!$, \ie 
\begin{equation}
\Pi _{k}\ =\ \big\{(\pi _{1},\pi _{2})\in S_{m}^{2}\ |\ \tau (\pi _{1},\pi
_{2})=\pi ^{[k]}\big\}\ =\ \big\{(\pi _{1},\pi _{2})\in S_{m}^{2}\ |\
\pi _{2}=\pi ^{[k]}\circ \pi _{1}\big\}.  \label{tr_pairs_k}
\end{equation}%
These are given by 
\begin{equation}
\begin{array}{@{}l}
\Pi _{1}=\Big\{(\pi ^{[1]},\pi ^{[1]}),(\pi ^{[2]},\pi
^{[2]}),(\pi ^{[3]},\pi ^{[3]}),(\pi ^{[4]},\pi
^{[4]}),(\pi ^{[5]},\pi ^{[5]}),(\pi ^{[6]},\pi
^{[6]})\Big\}, \\ 
\Pi _{2}=\Big\{(\pi ^{[1]},\pi ^{[2]}),(\pi ^{[2]},\pi
^{[1]}),(\pi ^{[3]},\pi ^{[4]}),(\pi ^{[4]},\pi
^{[3]}),(\pi ^{[5]},\pi ^{[6]}),(\pi ^{[6]},\pi
^{[5]})\Big\}, \\ 
\Pi _{3}=\Big\{(\pi ^{[1]},\pi ^{[3]}),(\pi ^{[2]},\pi
^{[5]}),(\pi ^{[3]},\pi ^{[1]}),(\pi ^{[4]},\pi
^{[6]}),(\pi ^{[5]},\pi ^{[2]}),(\pi ^{[6]},\pi
^{[4]})\Big\}, \\ 
\Pi _{4}=\Big\{(\pi ^{[1]},\pi ^{[4]}),(\pi ^{[2]},\pi
^{[6]}),(\pi ^{[3]},\pi ^{[2]}),(\pi ^{[4]},\pi
^{[5]}),(\pi ^{[5]},\pi ^{[1]}),(\pi ^{[6]},\pi
^{[3]})\Big\}, \\ 
\Pi _{5}=\Big\{(\pi ^{[1]},\pi ^{[5]}),(\pi ^{[2]},\pi
^{[3]}),(\pi ^{[3]},\pi ^{[6]}),(\pi ^{[4]},\pi
^{[1]}),(\pi ^{[5]},\pi ^{[4]}),(\pi ^{[6]},\pi
^{[2]})\Big\}, \\ 
\Pi _{6}=\Big\{(\pi ^{[1]},\pi ^{[6]}),(\pi ^{[2]},\pi
^{[4]}),(\pi ^{[3]},\pi ^{[5]}),(\pi ^{[4]},\pi
^{[2]}),(\pi ^{[5]},\pi ^{[3]}),(\pi ^{[6]},\pi
^{[1]})\Big\}.%
\end{array}
\label{tr_pairs}
\end{equation}

For all $\pi _{1},\pi_{2} \in S_{m}$, it holds $\tau (\pi _{2},\pi_{1})=\tau (\pi _{1},\pi_{2})^{-1}$. Due to equation \eqref{order2}, both $(\pi _{1},\pi_{2})$ and $(\pi _{2},\pi_{1})$ belong to the same set $\Pi _{k}$ according to equation \eqref{tr_pairs} if $k=1,2,3,6$.

\medskip
Furthermore, the Cayley distance $d_{\textup{C}}(\pi _{1},\pi _{2})$ in $S_{3}$ is
given by the matrix (see \citealp{amigo25}) 

\ba
\label{dC matrix}
\begin{array}{c|cccccc}
\toprule
\text{Distance $d_{\textup{C}}$ } & \pi^{[1]} & \pi^{[2]} & \pi^{[3]} & \pi^{[4]} & \pi^{[5]} & \pi^{[6]} \\
\midrule
\pi^{[1]} = (1,2,3) & 0 & 1 & 1 & 2 & 2 & 1 \\
\pi^{[2]} = (1,3,2) & 1 & 0 & 2 & 1 & 1 & 2 \\
\pi^{[3]} = (2,1,3) & 1 & 2 & 0 & 1 & 1 & 2 \\
\pi^{[4]} = (2,3,1) & 2 & 1 & 1 & 0 & 2 & 1 \\
\pi^{[5]} = (3,1,2) & 2 & 1 & 1 & 2 & 0 & 1 \\
\pi^{[6]} = (3,2,1) & 1 & 2 & 2 & 1 & 1 & 0 \\
\bottomrule
\end{array}
\ea

while the Kendall distance $d_{\textup{K}}(\pi _{1},\pi _{2})$ in $S_{3}$ is given by
the matrix (see \citealp{amigo25}) 

\ba
\label{dK matrix}
\begin{array}{c|cccccc}
\toprule
\text{Distance $d_{\textup{K}}$ } & \pi^{[1]} & \pi^{[2]} & \pi^{[3]} & \pi^{[4]} & \pi^{[5]} & \pi^{[6]} \\
\midrule
\pi^{[1]} = (1,2,3) & 0 & 1 & 1 & 2 & 2 & 3 \\
\pi^{[2]} = (1,3,2) & 1 & 0 & 2 & 3 & 1 & 2 \\
\pi^{[3]} = (2,1,3) & 1 & 2 & 0 & 1 & 3 & 2 \\
\pi^{[4]} = (2,3,1) & 2 & 3 & 1 & 0 & 2 & 1 \\
\pi^{[5]} = (3,1,2) & 2 & 1 & 3 & 2 & 0 & 1 \\
\pi^{[6]} = (3,2,1) & 3 & 2 & 2 & 1 & 1 & 0 \\
\bottomrule
\end{array}
\ea

\begin{bem}
\label{bemCayleyKendall}
The Cayley (resp.\ Kendall) distance allows $S_{m}$ to be represented by its adjacency graphs, \ie a graph where every node $\pi
=(i_{1},i_{2},\ldots,i_{m})\in S_{m}$ is connected to its $\frac{(m-1)}{2}$ (resp.\ $m-1$) nearest neighbors, namely those permutations that differ from  $\pi $ (viewed as symbolic strings) due to transpositions (resp.\ adjacent transpositions) of two symbols $i_{j},i_{k}$ for $1\leq j<k\leq m$ (resp.\ $i_{j},i_{j+1}$ for $1\leq j\leq m-1
$). Figure~\ref{Graphs} shows the Cayley and Kendall adjacency graphs of $S_{3}$. In other words, Figure~\ref{Graphs} (a) and (b) visualize the distance matrices \eqref{dC matrix} and \eqref{dK matrix}, respectively. The fact that the Cayley graph has edge crossings while the Kendall graph has not is because $d_{\textup{K}}$ involves only adjacent transpositions while $d_{\textup{C}}$ allows also non-adjacent ones. For the Kendall adjacency graph of $S_{4}$, see \citet[Figure~2]{amigo25}.
\end{bem}

\begin{figure}[th]
\centering
(a)\includegraphics[viewport=15 20 125 130, clip=, scale=1]{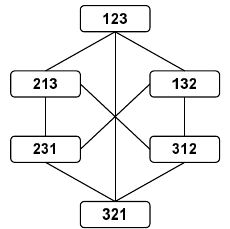}
\qquad
(b)\includegraphics[viewport=15 20 125 130, clip=, scale=1]{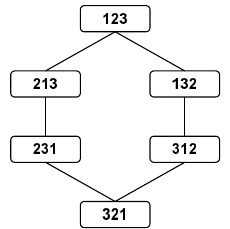}
\caption{The Cayley (a) and Kendall (b) adjacency graphs of the symmetric group $S_{3}$, nodes corresponding to permutations $(i_{1},i_{2},i_{3})$ written as symbolic strings $i_{1}i_{2}i_{3}$. All nodes in panel (a) have valency~3 and the edges between them correspond to Cayley distance~1 (\ie they differ in a single transposition). All nodes in panel (b) have valency~2 and the edges between them have Kendall distance~1 (\ie they differ in a single adjacent transposition). }
\label{Graphs}
\end{figure}

\smallskip
According to \citet{amigo25}, the right-invariance 
\begin{equation}
d_{\textup{C},\textup{K}}(\pi _{1},\pi _{2})=d_{\textup{C},\textup{K}}(\mathrm{id},\pi _{2}\circ\pi_{1}^{-1})=d_{\textup{C},\textup{K}}\big(\mathrm{id},\tau (\pi _{1},\pi _{2})\big)%
\text{,}  \label{d(p1,p2) invariance}
\end{equation}%
holds for all $\pi _{1},\pi _{2}\in S_{m}$, where $d_{\textup{C},\textup{K}}$ stands for both the
Cayley distance $d_{\textup{C}}$ and the Kendall distance $d_{\textup{K}}$, and $\mathrm{id}$
denotes the identity permutation $(1,2,\ldots,m-1,m)$. By equation \eqref{d(p1,p2) invariance}, the distance $d_{\textup{C},\textup{K}}(\pi _{1},\pi _{2})$ in the symmetric group $S_{m}$ can be read in
the $\pi_{1}=\mathrm{id}$ row of the corresponding distance
matrix $(d_{\textup{C},\textup{K}}\big(\pi _{1},\pi _{2})\big)_{\pi _{1},\pi _{2}\in S_{m}}$. In the
case $m=3$, the first rows of the matrices \eqref{dC matrix} and \eqref{dK matrix} spell out the following relationships between $d_{\textup{C},\textup{K}}(\pi _{1},\pi _{2})$ and the sets $\Pi_{k}$ in equation \eqref{tr_pairs}.

\begin{prop}
\label{PropDistancesCK}
It holds:
\begin{description}
\item[(a)] $d_{\textup{C},\textup{K}}(\pi _{1},\pi _{2})=0$ iff $\tau (\pi _{1},\pi
_{2})=\pi ^{[1]}$, \ie $(\pi _{1},\pi _{2})\in \Pi _{1}$.

\item[(b)] $d_{\textup{C}}(\pi _{1},\pi _{2})=1$ iff $\tau (\pi _{1},\pi
_{2})\in \{\pi ^{[2]},\pi ^{[3]},\pi ^{[6]}\}$, \ie $%
(\pi _{1},\pi _{2})\in \Pi _{2}\cup \Pi _{3}\cup \Pi _{6}$.

\item[(c)] $d_{\textup{K}}(\pi _{1},\pi _{2})=1$ iff $\tau (\pi _{1},\pi
_{2})\in \{\pi ^{[2]},\pi ^{[3]}\}$, \ie $(\pi _{1},\pi
_{2})\in \Pi _{2}\cup \Pi _{3}$.

\item[(d)] $d_{\textup{C},\textup{K}}(\pi _{1},\pi _{2})=2$ iff $\tau (\pi _{1},\pi
_{2})\in \{\pi ^{[4]},\pi ^{[5]}\}$, \ie $(\pi _{1},\pi
_{2})\in \Pi _{4}\cup \Pi _{5}$.

\item[(e)] $d_{\textup{K}}(\pi _{1},\pi _{2})=3$ iff $\tau (\pi _{1},\pi
_{2})=\pi ^{[6]}$, \ie $(\pi _{1},\pi _{2})\in \Pi _{6}$.
\end{description}
\end{prop}

\begin{bem}
\label{bemPeriod}
The period of $\pi \in S_{m}$ (or an element of any other group, for this matter) is the minimum positive integer 
\textrm{per}$(\pi )$ such that $\pi ^{\mathrm{per}(\pi )}=\mathrm{id}$, the identity element of the group. It follows that $\mathrm{id}$ is the only element of period~1 (a ``fixed point''), and \textrm{per}$(\pi )$ is a divisor of $m!$, the cardinality of the group $S_{m}$ \citep{herstein96}. All permutations of a given period $p$ build a so-called \textit{period class} $\mathcal{C}_{p}$, \ie $\mathcal{C}_{p}=\{\pi \in S_{m}:\mathrm{per}(\pi )=p\}$. If $m=3$, then we have (see \eqref{tabOP})%
\begin{equation}
\mathcal{C}_{1}=\{\pi ^{[1]}\},\;\mathcal{C}_{2}=\{\pi ^{\lbrack
2]},\pi ^{[3]},\pi ^{[6]}\},\;\mathcal{C}_{3}=\{\pi ^{\lbrack
4]},\pi ^{[5]}\}.  \label{order classes}
\end{equation}%
Comparison of equation \eqref{order classes} to Proposition~\ref{PropDistancesCK} demonstrates that 
\begin{equation}
d_{\textup{C}}(\pi _{1},\pi _{2})=n\;\;\text{iff}\;\;\tau (\pi _{1},\pi _{2})\in 
\mathcal{C}_{n+1},  \label{order classes 2}
\end{equation}%
where $n=0,1,2$. Therefore, $d_{\textup{C}}(\pi _{1},\pi _{2})=\mathrm{per}\big(\tau (\pi
_{1},\pi _{2})\big)-1$ for $S_{3}$, so that the
period of $\tau (\pi _{1},\pi _{2})$, decreased by $1$, is the Cayley distance between $\pi _{1}$ and $\pi _{2}$ in $S_{3}$.
\end{bem}


\section{Distributional Properties of Transcripts}
\label{Distributional Properties of Transcripts}
From now on, we assume $(X_t)$ to be a real-valued and continuously distributed process having a stationary OP-series $(\pi_t)$ of length $\cat=3$. The corresponding transcript series with coupling delay $\Lambda=1$, so $(\pi_{t+1}\circ\pi_{t}^{-1})$ according to \eqref{self and cross}, is denoted by $(\tau_t)$, and the resulting series of Cayley and Kendall distances by $(d_{\textup{C},t})$ and $(d_{\textup{K},t})$, respectively, where $d_{\textup{C},t}=d_{\textup{C}}(\pi_{t},\pi_{t+1})$ and $d_{\textup{K},t}=d_{\textup{K}}(\pi_{t},\pi_{t+1})$.


\subsection{Marginal Distribution of Transcripts}
\label{Marginal Distribution of Transcripts}
In order to compute the stationary marginal distribution of the transcripts~$(\tau_t)$, \ie the vector $\fp_{\tau} = \big(p_{\tau;1}, \ldots, p_{\tau;m!}\big)^\top$, the joint bivariate distribution of the OP-pairs $(\pi_t,\pi_{t+1})$ is required, because
\ba
\label{tr_marg}
p_{\tau;k} := P\big(\tau_t=\pi^{[k]}\big)\ =\ \sum_{(a,b)\in\Pi_k} P(\pi_t=a, \pi_{t+1}=b),
\ea
where the summation is done across the sets in \eqref{tr_pairs}. Let us introduce the short-hand notations $p_{\pi;k} = P(\pi_t=\pi^{[k]})$ and $\fp_\pi=(p_{\pi;1},\ldots,p_{\pi;m!})^\top$ for the OPs' marginal distribution, as well as $p_{\pi;kl}(h) = P(\pi_t=\pi^{[k]}, \pi_{t+h}=\pi^{[l]})$ and $\mP\!_\pi(h)=\big(p_{\pi;kl}(h)\big)_{k,l=1,\ldots,\cat!}$ for the (matrix of all) probabilities at \textit{time lag} $h\in\bbn$, with $\mP\!_\pi(-h) = \mP\!_\pi(h)^\top$. It should be noted, however, that for consecutive OPs (lag $h=1$), some transitions are impossible, irrespective of the underlying process $(X_t)$. As summarized in Figure~1 of \citet{sousa22}, it is impossible that $\pi^{[k]}$ with $k\in\{1,3,4\}$ is followed by $\pi^{[l]}$ with $l\in\{3,4,6\}$, and analogously with $k\in\{2,5,6\}$ and $l\in\{1,2,5\}$, where we took into account that \citet{sousa22} use different rules than ours to define and order OPs. Hence, the bivariate probabilities~$p_{\pi;kl}(1)$ corresponding to such impossible transitions are equal to zero. Therefore, using \eqref{tr_pairs} and \eqref{tr_marg}, it generally holds for $h=1$ that
\ba
\label{tr_marg2}
\begin{array}{@{}rl}
p_{\tau;1} = P\big(\tau_t=\pi^{[1]}\big)\ =& 
p_{\pi;11}(1) + p_{\pi;66}(1),\\[1ex]
p_{\tau;2} = P\big(\tau_t=\pi^{[2]}\big)\ =& 
p_{\pi;12}(1) + p_{\pi;56}(1),\\[1ex]
p_{\tau;3} = P\big(\tau_t=\pi^{[3]}\big)\ =& 
p_{\pi;31}(1) + p_{\pi;64}(1),\\[1ex]
p_{\tau;4} = P\big(\tau_t=\pi^{[4]}\big)\ =& 
p_{\pi;26}(1) + p_{\pi;32}(1) + p_{\pi;45}(1) + p_{\pi;63}(1),\\[1ex]
p_{\tau;5} = P\big(\tau_t=\pi^{[5]}\big)\ =& 
p_{\pi;15}(1) + p_{\pi;23}(1) + p_{\pi;41}(1) + p_{\pi;54}(1),\\[1ex]
p_{\tau;6} = P\big(\tau_t=\pi^{[6]}\big)\ =& 
p_{\pi;24}(1) + p_{\pi;35}(1) + p_{\pi;42}(1) + p_{\pi;53}(1).
\end{array}
\ea
Altogether, each transcript corresponds to a set of possible 4-OPs, see Table~\ref{tabTrLag0} for a summary.

\begin{table}[t]
\centering\footnotesize
\caption{Possible 4-OPs obtained if a certain transcript is observed.}
\label{tabTrLag0}

\smallskip
\begin{tabular}{cccccc}
\toprule
$\tau_t=\pi^{[1]}$ & $\tau_t=\pi^{[2]}$ & $\tau_t=\pi^{[3]}$ & $\tau_t=\pi^{[4]}$ & $\tau_t=\pi^{[5]}$ & $\tau_t=\pi^{[6]}$ \\
 \midrule
$(1,2,3,4)$ & $(1,2,4,3)$ & $(2,1,3,4)$ & $(1,4,3,2)$ & $(1,3,2,4)$ & $(1,3,4,2)$ \\
$(4,3,2,1)$ & $(4,3,1,2)$ & $(3,4,2,1)$ & $(2,1,4,3)$ & $(1,4,2,3)$ & $(2,4,3,1)$ \\
 &  &  & $(2,4,1,3)$ & $(2,3,1,4)$ & $(3,1,2,4)$ \\
 &  &  & $(3,2,1,4)$ & $(2,3,4,1)$ & $(4,2,1,3)$ \\
 &  &  & $(3,2,4,1)$ & $(3,1,4,2)$ &  \\
 &  &  & $(4,1,3,2)$ & $(3,4,1,2)$ &  \\
 &  &  & $(4,2,3,1)$ & $(4,1,2,3)$ &  \\
 \bottomrule
 \end{tabular}
\end{table}

\smallskip
Furthermore, using Proposition~\ref{PropDistancesCK} and equation \eqref{tr_marg2}, it follows that the marginal distributions of Cayley and Kendall distances, 
$$
\fp_{\textup{C}} = (p_{\textup{C};0},\ldots,p_{\textup{C};2})^\top := \big(P( d_{\textup{C},t}=0), \ldots, P( d_{\textup{C},t}=2)\big)^\top
$$
and 
$$
\fp_{\textup{K}} = (p_{\textup{K};0},\ldots,p_{\textup{K};3})^\top := \big(P( d_{\textup{K},t}=0), \ldots, P( d_{\textup{K},t}=3)\big)^\top, 
$$
respectively, are obtained by multiplying $\fp_{\tau}$ with the transformation matrices
\ba
\label{trans_mat_C}
\mT_{\textup{C}}\ =\ \left(\begin{array}{cccccc}
1 & 0 & 0 & 0 & 0 & 0 \\
0 & 1 & 1 & 0 & 0 & 1 \\
0 & 0 & 0 & 1 & 1 & 0 \\
\end{array}\right)
\ea
and
\ba
\label{trans_mat_K}
\mT_{\textup{K}}\ =\ \left(\begin{array}{cccccc}
1 & 0 & 0 & 0 & 0 & 0 \\
0 & 1 & 1 & 0 & 0 & 0 \\
0 & 0 & 0 & 1 & 1 & 0 \\
0 & 0 & 0 & 0 & 0 & 1 \\
\end{array}\right),
\ea
respectively, \ie

\begin{equation}
\fp_{\textup{C}}\ =\ \mT_{\textup{C}}\, \fp_{\tau},
\qquad
\fp_{\textup{K}}\ =\ \mT_{\textup{K}}\, \fp_{\tau}.
\label{p_C=T_C*p_tau}
\end{equation}

\bigskip
The lag-1 probabilities required for \eqref{tr_marg2} have already been computed for important classes of process $(X_t)$. In view of a possible application to dependence testing, the case of an \iid\ process $(X_t)$ is relevant. Recall that $(X_t)$ is continuously distributed, so under \iid-conditions, it holds that $(\pi_t)$ has a discrete-uniform marginal distribution, $\fp_\pi=(\frac{1}{m!}, \ldots, \frac{1}{m!})^\top$, irrespective of the marginal distribution of $(X_t)$. This distribution-free property is very attractive for practice, because OP-based tests for serial dependence become nonparametric. If $(X_t)$ is \iid, then the lag-1 OP-probabilities are \citep[see][]{elsinger10,sousa22,weiss22}
\ba
\label{op_biv1_iid}
\mP\!_\pi(1)\ =\ 
\frac{1}{24}\,\left(\begin{array}{cccccc}
1 & 1 & 0 & 0 & 2 & 0 \\
0 & 0 & 1 & 1 & 0 & 2 \\
1 & 2 & 0 & 0 & 1 & 0 \\
2 & 1 & 0 & 0 & 1 & 0 \\
0 & 0 & 1 & 2 & 0 & 1 \\
0 & 0 & 2 & 1 & 0 & 1 \\
\end{array}\right).
\ea
Note that each of these authors uses a different order of the OPs, which also differs from the lexicographic order being used here. Using equations \eqref{tr_marg2}--\eqref{trans_mat_K} together with equation \eqref{op_biv1_iid}, we obtain the following result.

\begin{prop}
\label{prop_tr_marg_iid}
If $(X_t)$ is \iid, then the transcripts' marginal distribution is given by
$$
\fp_{\tau}\ =\ 
\tfrac{1}{24}\,\big(2, 2, 2, 7, 7, 4\big)^\top.
$$
Furthermore, the Cayley and Kendall distances of successive OPs satisfy
$$
\begin{array}{@{}l}
\fp_{\textup{C}}\ =\ 
\tfrac{1}{12}\,\big(1, 4, 7\big)^\top\quad
\text{ with mean } \mu_{\textup{C}}=\tfrac{3}{2}
\text{ and variance }
\sigma_{\textup{C}}^2=\tfrac{5}{12},
\\[2ex]
\fp_{\textup{K}}\ =\ 
\tfrac{1}{12}\,\big(1, 2, 7, 2\big)^\top\quad
\text{ with mean } \mu_{\textup{K}}=\tfrac{11}{6}
\text{ and variance }
\sigma_{\textup{K}}^2=\tfrac{23}{36}.
\end{array}
$$
\end{prop}
Note that in the \iid-case, $\fp_{\tau}$ simply agrees with the relative frequencies of the respective 4-OPs in Table~\ref{tabTrLag0}. The \iid-case is particularly relevant for practice, namely if one wants to test the null hypothesis of serial independence against dependent alternatives (as we do later in Sections~\ref{Asymptotics of Transcript Statistics}--\ref{Empirical Investigations}). 

\smallskip
But in view of power properties of such dependence tests, it is also relevant to know the transcripts' marginal distribution for serially dependent processes $(X_t)$. Unfortunately, closed-form expressions for OP-distributions are known for very few types of process so far. In what follows, we pick up two alternative scenarios where closed-form OP-distributions are readily available. First, we consider the case of $(X_t)$ being a random walk with symmetric noise (SRW). Then, see equation (42) in \citet{sousa22}, we have
\ba
\label{op_biv1_srw}
\mP\!_\pi(1)\ =\ 
\frac{1}{48}\,\left(\begin{array}{cccccc}
6 & 3 & 0 & 0 & 3 & 0 \\
0 & 0 & 2 & 1 & 0 & 3 \\
3 & 2 & 0 & 0 & 1 & 0 \\
3 & 1 & 0 & 0 & 2 & 0 \\
0 & 0 & 1 & 2 & 0 & 3 \\
0 & 0 & 3 & 3 & 0 & 6 \\
\end{array}\right).
\ea
So using equations \eqref{tr_marg2}--\eqref{trans_mat_K} together with equation \eqref{op_biv1_srw}, we obtain the following result.

\begin{prop}
\label{prop_tr_marg_srw}
If $(X_t)$ is an SRW, then the transcripts' marginal distribution is given by
$$
\fp_{\tau}\ =\ 
\tfrac{1}{24}\,\big(6, 3, 3, 5, 5, 2\big)^\top.
$$
Furthermore, the Cayley and Kendall distances of successive OPs satisfy
$$
\begin{array}{@{}l}
\fp_{\textup{C}}\ =\ 
\tfrac{1}{12}\,\big(3, 4, 5\big)^\top\quad
\text{ with mean } \mu_{\textup{C}}=\tfrac{7}{6}
\text{ and variance }
\sigma_{\textup{C}}^2=\tfrac{23}{36},
\\[2ex]
\fp_{\textup{K}}\ =\ 
\tfrac{1}{12}\,\big(3, 3, 5, 1\big)^\top\quad
\text{ with mean } \mu_{\textup{K}}=\tfrac{4}{3}
\text{ and variance }
\sigma_{\textup{K}}^2=\tfrac{8}{9}.
\end{array}
$$
\end{prop}
As another alternative, we consider the generalized coin-tossing (GCT) process of \citet{silbernagel26} with ``success'' probability $p = 1-q \in (0;1)$, which covers the ``fair'' CT process of \citet{bandt25} for $p=q=0.5$. Its bivariate lag-1 probabilities are given by
\ba
\label{op_biv1_gct}
\mP\!_\pi(1)\ =\ 
\left(\begin{array}{cccccc}
p^3 & p^3 q & 0 & 0 & p^2 q^2 & 0 \\
0 & 0 & p^4 q & p^3 q^2 & 0 & p^2 q^2 \\
p^3 q & p^3 q^2 & 0 & 0 & p^2 q^3 & 0 \\
p^2 q^2 & p^2 q^3 & 0 & 0 & p q^4 & 0 \\
0 & 0 & p^3 q^2 & p^2 q^3 & 0 & p q^3 \\
0 & 0 & p^2 q^2 & p q^3 & 0 & q^3 \\
\end{array}\right),
\ea
see \citet[p.~17]{silbernagel26}. So with an analogous argumentation as before, we obtain the following result.

\begin{prop}
\label{prop_tr_marg_gct}
If $(\pi_t)$ is a GCT process, then the transcripts' marginal distribution is given by
$$
\fp_{\tau}\ =\ 
\Big(p^3 + q^3,\ pq(p^2 + q^2),\ pq(p^2 + q^2),\ p q^2 (1 + 2 p^2),\ p^2 q (1 + 2 q^2),\ 2 p^2 q^2\Big)^\top.
$$
Furthermore, the Cayley and Kendall distances of successive OPs satisfy
$$
\fp_{\textup{C}}\ =\ 
\Big(p^3+q^3,\ 2 pq(p^2 + q),\ p q (1 + 2 p q)\Big)^\top
$$
with mean $\mu_{\textup{C}}=2\, pq\, (2 + pq)$ and variance $\sigma_{\textup{C}}^2=2\, p q \, (3+p q)(1-2 p q-2 p^2 q^2)$, and
$$
\fp_{\textup{K}}\ =\ 
\Big(p^3+q^3,\ 2 pq(p^2 + q^2),\ p q (1 + 2 p q),\ 2 p^2 q^2\Big)^\top
$$
with mean $\mu_{\textup{K}}=2\, pq\, (2 + 3 pq)$ and variance $\sigma_{\textup{K}}^2=6\, p q\, (1+3 p q) (1-2 p q-2 p^2 q^2)$, respectively.
\end{prop}

\begin{proof}
Using \eqref{tr_marg2} together with \eqref{op_biv1_iid}, it follows that 
$$
\begin{array}{@{}rl}
p_{\tau;1}\ =& 
p_{\pi;11}(1) + p_{\pi;66}(1)\ =\ p^3 + q^3,\\[1ex]
p_{\tau;2}\ =& 
p_{\pi;12}(1) + p_{\pi;56}(1)
\ =\ p^3 q + p q^3\ =\ pq(p^2 + q^2),\\[1ex]
p_{\tau;3}\ =& 
p_{\pi;31}(1) + p_{\pi;64}(1)
\ =\ p^3 q + p q^3\ =\ pq(p^2 + q^2),\\[1ex]
p_{\tau;4}\ =& 
p_{\pi;26}(1) + p_{\pi;32}(1) + p_{\pi;45}(1) + p_{\pi;63}(1)
\\
=&
p^2 q^2 + p^3 q^2 + p q^4 + p^2 q^2
\ =\ p q^2 (1 + 2 p^2),\\[1ex]
p_{\tau;5}\ =& 
p_{\pi;15}(1) + p_{\pi;23}(1) + p_{\pi;41}(1) + p_{\pi;54}(1)
\\
=&
p^2 q^2 + p^4 q + p^2 q^2 + p^2 q^3
\ =\ p^2 q (1 + 2 q^2),\\[1ex]
p_{\tau;6}\ =& 
p_{\pi;24}(1) + p_{\pi;35}(1) + p_{\pi;42}(1) + p_{\pi;53}(1)
\\
=&
p^3 q^2 + p^2 q^3 + p^2 q^3 + p^3 q^2
\ =\ 2 p^2 q^2.
\end{array}
$$
From equations \eqref{trans_mat_C}--\eqref{trans_mat_K}, it then follows that
$$
\begin{array}{@{}l}
p_{\textup{C};0} = p_{\textup{K};k} = P\big(\tau_t=\pi^{[1]}\big) = p^3+q^3;
\\[1ex]
p_{\textup{C};1} = P\big(\tau_t \in \{\pi^{[2]}, \pi^{[3]}, \pi^{[6]}\}\big) 
= 2 pq(p^2 + q);
\\[1ex]
p_{\textup{K};1} = P\big(\tau_t \in \{\pi^{[2]}, \pi^{[3]}\}\big) = 2 pq(p^2 + q^2);
\\[1ex]
p_{\textup{C};2} = p_{\textup{K};2} = P\big(\tau_t \in \{\pi^{[4]}, \pi^{[5]}\}\big)
= p q (1 + 2 p q);
\\[1ex]
p_{\textup{K};3} = P\big(\tau_t=\pi^{[6]}\big) = 2 p^2 q^2.
\end{array}
$$
As a consequence, the respective mean distances are computed as
$$
\begin{array}{@{}rl}
\mu_{\textup{C}}=E[d_{\textup{C},t}]\ =&
p_{\textup{C};1} + 2\, p_{\textup{C};2}
\ =\ 
2 pq\,\big(p^2 + q + 1 + 2 p q\big)
\ =\ 
2 pq(2 + p q),
\\[1ex]
\mu_{\textup{K}}=E[d_{\textup{K},t}]\ =&
p_{\textup{K};1} + 2\, p_{\textup{K};2} + 3\, p_{\textup{K};3}
\\
=&
2 pq\,\big(p^2 + q^2 + 1 + 2 p q + 3 pq\big)
\ =\ 
2 pq (2 + 3 pq),
\end{array}
$$
and the variances are computed analogously. So the proof of Proposition~\ref{prop_tr_marg_gct} is complete.
\end{proof}
Note that both expected distances in Proposition~\ref{prop_tr_marg_gct} are quadratic polynomials in~$pq$ (both strictly increasing for $pq\in [0;0.25]$), where~$pq$ itself is a quadratic polynomial in~$p$ that is maximized for $p=0.5$ (fair coin tossing). Similarly, the variances are quartic polynomials in~$pq$. The moment properties of Propositions~\ref{prop_tr_marg_iid}--\ref{prop_tr_marg_gct} are illustrated by Figure~\ref{figExpectedDistances}. Looking at the mean distances in the upper panel (parts~(a) and~(b)), it gets clear that the \iid-case leads to a much larger mean distance (both for $d_{\textup{C}}$ and~$d_{\textup{K}}$) than the serially dependent SRW and GCT process. This implies that both mean distances could be used for constructing test statistics in order to test the null hypothesis of serial independence against dependence, a topic that shall be investigated in more detail in Section~\ref{Asymptotics of Transcript Statistics} below. For the variances in the lower panel of Figure~\ref{figExpectedDistances} (parts~(c) and~(d)), it is interesting to note that the \iid-case often leads to the lowest value, except for a GCT process with very low or very high~$p$.

\begin{figure}[t]
\centering\footnotesize
(a)\hspace{-3ex}%
\includegraphics[viewport=0 45 265 235, clip=, scale=0.6]{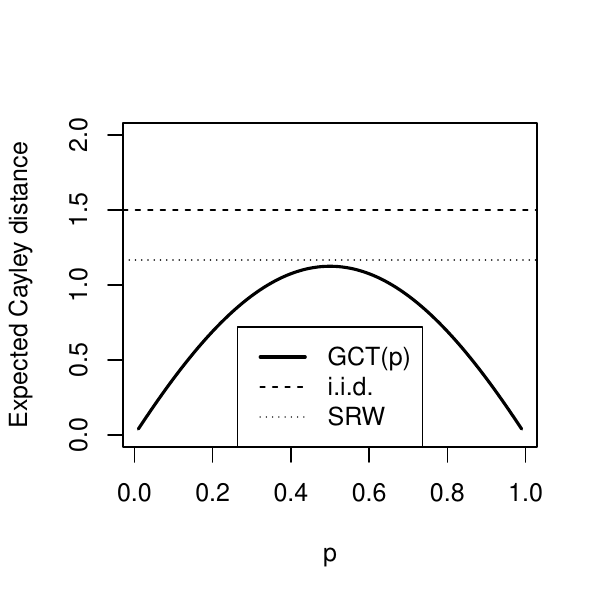}\,$p$
\qquad
(b)\hspace{-3ex}%
\includegraphics[viewport=0 45 265 235, clip=, scale=0.6]{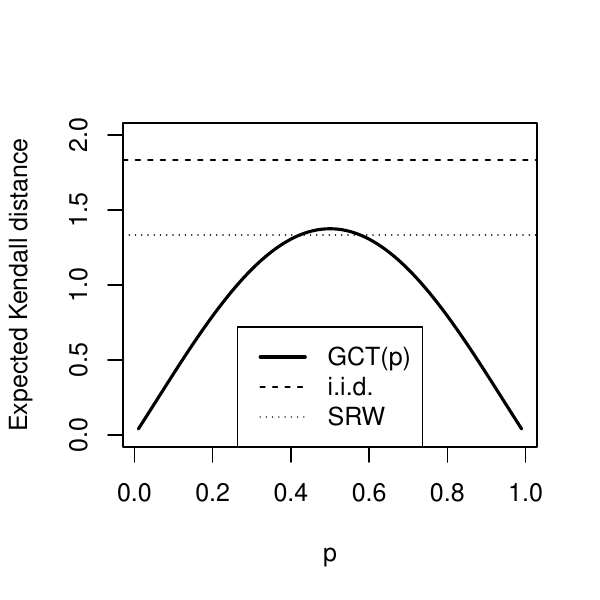}\,$p$
\\[2ex]
(c)\hspace{-3ex}%
\includegraphics[viewport=0 45 265 235, clip=, scale=0.6]{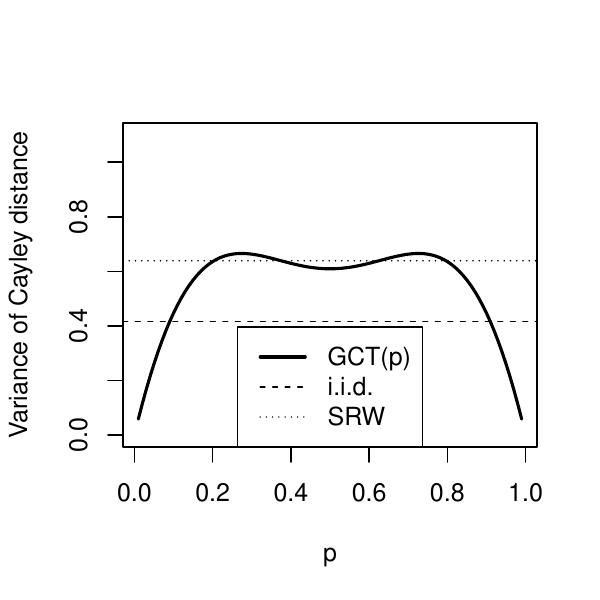}\,$p$
\qquad
(d)\hspace{-3ex}%
\includegraphics[viewport=0 45 265 235, clip=, scale=0.6]{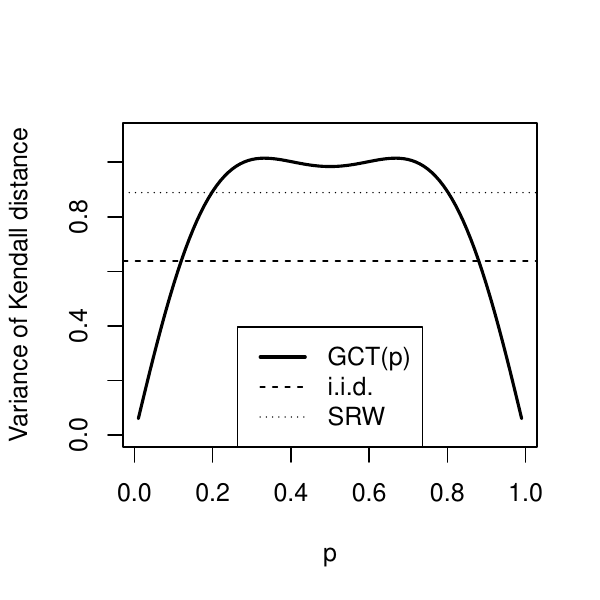}\,$p$
\caption{Expected Cayley distance in (a) and Kendall distance in (b) for GCT process, plotted against success probability~$p$. Lines indicate corresponding values for \iid\ process or SRW, respectively. Parts~(c) and~(d) show the corresponding variances.}
\label{figExpectedDistances}
\end{figure}

\begin{bsp}
\label{bspAlternativeScenarios}
Let us investigate the distributional and moment properties of transcripts, Cayley and Kendall distance for further types of process $(X_t)$. In our simulation study in Section~\ref{Simulation-based Performance Analyses} below, we shall compare the null hypothesis of serial independence (\ie the \iid-case of Proposition~\ref{prop_tr_marg_iid}) to various alternative scenarios with rather different serial dependence structures (mostly taken from \citealp{weiss22}), namely the
\begin{itemize}
    \item first-order autoregressive (AR$(1)$) process $X_t=\phi\,X_{t-1}+\epsilon_t$ with $\phi\in\{-0.5,0.5\}$ and \iid\ normal $\epsilon_t\sim \norm(0,1)$;
    \item first-order quadratic moving-average (QMA$(1)$) process $X_t=\epsilon_t + 0.8\,\epsilon_{t-1}^2$ with \iid\ $\epsilon_t\sim \norm(0,1)$;
    \item first-order transposed exponential autoregressive TEAR$(1)$ process $X_t=B_t\,X_{t-1}+0.85\,\epsilon_t$ with \iid\ exponential $\epsilon_t\sim \textup{Exp}(1)$ and Bernoulli $B_t$ with $P(B_t=1)=0.15$;
    \item first-order autoregressive conditional heteroscedasticity ARCH$(1)$ process $X_t = \epsilon_t\,\sqrt{0.2+0.8\,X_{t-1}^2}$ with \iid\ $\epsilon_t\sim \norm(0,1)$.
\end{itemize}
Since an analytical computation of the relevant distributional and moment properties is not possible for these processes (as closed-form OP-distributions have not been derived so far), we approximate them by their empirical counterparts being computed from a simulated time series of length~$10^6$. The obtained results are summarized in Table~\ref{tabDistMomTranscripts}.

\smallskip
The tabulated properties for both AR$(1)$ processes deviate considerably from those of the \iid-case. Hence, if we later construct tests of the \iid-null, we expect a good power performance for AR$(1)$ alternatives. It is interesting to note that the direction of these deviations strongly depends on the AR~parameter's sign. For example, the mean Cayley resp.\ Kendall distance is decreased (increased) if the AR$(1)$ process is positively (negatively) dependent. This implies that a test based on the mean distance should be constructed in a two-sided manner, \ie it should be able to detect both increases and decreases in the mean distance. Furthermore, the probability $p_{\tau; 1} = p_{\textup{C}; 0} = p_{\textup{K}; 0}$ increases (decreases) for positive (negative) dependence, which can be explained as follows. Recall that transcripts are a kind of ``difference'' between successive OPs, and $\tau_t=\pi^{[1]}$ (equivalently, $d_{\textup{C},t} = d_{\textup{K},t} = 0$) corresponds to the repeated occurrence of the same OP (so no difference at all), see \eqref{tr_pairs}--\eqref{dK matrix}. More precisely, it corresponds to the repeated occurrence of either~$\pi^{[1]}$ or~$\pi^{[6]}$, see \eqref{tr_marg2}. In fact, positive dependence is characterized by long lasting rises (then we observe~$\pi^{[1]}$) or descents (then~$\pi^{[6]}$), while the opposite happens for negative dependence.

\smallskip
For the remaining three types of process, the deviations to the \iid-case are generally less pronounced, so we conjecture that the power of dependence tests is lower than for the AR$(1)$ processes. Nevertheless, there are visible deviations (\eg for~$p_{\tau; 4}$ and~$p_{\tau; 5}$ in case of the QMA$(1)$ and TEAR$(1)$ process, or for $\mu_{\textup{C}}$ in case of the ARCH$(1)$ process), so appropriately designed dependence tests should still be able to uncover these alternatives.
\end{bsp}

\begin{table}[t]
\centering\small
\caption{Distributional and moment properties of (a) transcripts, (b) Cayley distance, and (c) Kendall distance for various types of process $(X_t)$, see Example~\ref{bspAlternativeScenarios}.}
\label{tabDistMomTranscripts}

\smallskip
(a)~\begin{tabular}{l|cccccc}
\toprule
Process & $p_{\tau; 1}$ & $p_{\tau; 2}$ & $p_{\tau; 3}$ & $p_{\tau; 4}$ & $p_{\tau; 5}$ & $p_{\tau; 6}$ \\
\midrule
\it \iid{} & \it 0.083 & \it 0.083 & \it 0.083 & \it 0.292 & \it 0.292 & \it 0.167 \\[1ex]
AR$(1),\ +0.5$ & 0.150 & 0.115 & 0.115 & 0.256 & 0.254 & 0.110 \\
AR$(1),\ -0.5$ & 0.041 & 0.042 & 0.042 & 0.293 & 0.293 & 0.287 \\
QMA$(1)$ & 0.103 & 0.076 & 0.094 & 0.348 & 0.220 & 0.160 \\
TEAR$(1)$ & 0.101 & 0.085 & 0.089 & 0.233 & 0.341 & 0.152 \\
ARCH$(1)$ & 0.100 & 0.085 & 0.085 & 0.279 & 0.279 & 0.172 \\
\bottomrule
\end{tabular}

\bigskip
(b)~\begin{tabular}{l|ccc@{\qquad}cc}
\toprule
Process & $p_{\textup{C}; 0}$ & $p_{\textup{C}; 1}$ & $p_{\textup{C}; 2}$ & $\mu_{\textup{C}}$ & $\sigma_{\textup{C}}^2$ \\
\midrule
\it \iid{} & \it 0.083 & \it 0.333 & \it 0.583 & \it 1.500 & \it 0.417 \\[1ex]
AR$(1),\ +0.5$ & 0.150 & 0.341 & 0.509 & 1.359 & 0.530 \\
AR$(1),\ -0.5$ & 0.041 & 0.372 & 0.587 & 1.545 & 0.331 \\
QMA$(1)$ & 0.103 & 0.330 & 0.568 & 1.465 & 0.454 \\
TEAR$(1)$ & 0.101 & 0.326 & 0.574 & 1.473 & 0.451 \\
ARCH$(1)$ & 0.100 & 0.342 & 0.558 & 1.458 & 0.448 \\
\bottomrule
\end{tabular}

\bigskip
(c)~\begin{tabular}{l|cccc@{\qquad}cc}
\toprule
Process & $p_{\textup{K}; 0}$ & $p_{\textup{K}; 1}$ & $p_{\textup{K}; 2}$ & $p_{\textup{K}; 3}$ & $\mu_{\textup{K}}$ & $\sigma_{\textup{K}}^2$ \\
\midrule
\it \iid{} & \it 0.083 & \it 0.167 & \it 0.583 & \it 0.167 & \it 1.833 & \it 0.639 \\[1ex]
AR$(1),\ +0.5$ & 0.150 & 0.231 & 0.509 & 0.110 & 1.580 & 0.764 \\
AR$(1),\ -0.5$ & 0.041 & 0.085 & 0.587 & 0.287 & 2.120 & 0.523 \\
QMA$(1)$ & 0.103 & 0.170 & 0.568 & 0.160 & 1.785 & 0.694 \\
TEAR$(1)$ & 0.101 & 0.174 & 0.574 & 0.152 & 1.777 & 0.679 \\
ARCH$(1)$ & 0.100 & 0.171 & 0.558 & 0.172 & 1.801 & 0.702 \\
\bottomrule
\end{tabular}
\end{table}


\subsection{Bivariate Distributions of Transcripts}
\label{Bivariate Distributions of Transcripts}
If dependence tests based on the aforementioned distributional or mean properties should be developed, we also require the bivariate distributions of transcripts in order to being able to derive the asymptotic distributions of test statistics. While the latter shall be done in Section~\ref{Asymptotics of Transcript Statistics} below, we start here by extending the computations of Section~\ref{Marginal Distribution of Transcripts} to joint bivariate distributions with time lag~$h\in\bbn$, although the necessary computations are much more cumbersome. For $\mP\!_\tau(h)=\big(p_{\tau;kl}(h)\big)_{k,l=1,\ldots,m!}$, it generally holds that
\ba
\label{tr_biv_h}
	\begin{array}{@{}l}
	p_{\tau;kl}(h) := P\big(\tau_t=\pi^{[k]}, \tau_{t+h}=\pi^{[l]}\big)\ =\ 
	\\[1ex]
	\displaystyle
    \qquad\sum_{(a,b)\in\Pi_k} \sum_{(\alpha,\beta)\in\Pi_l} P\Big(\pi_t=a, \pi_{t+1}=b,\ \pi_{t+h}=\alpha, \pi_{t+h+1}=\beta\Big),
	\end{array}
\ea
where the summations are done across the sets in \eqref{tr_pairs}. Also recall that $\mP\!_\tau(-h) = \mP\!_\tau(h)^\top$. Here, one has to consider again the impossible transitions between consecutive OPs~$\pi^{[i]}$ and~$\pi^{[j]}$, namely $\{1,3,4\}\not\to\{3,4,6\}$ and $\{2,5,6\}\not\to\{1,2,5\}$. Furthermore, for time lag $h=1$, the double summation in \eqref{tr_biv_h} is further constrained by the condition $b=\alpha$, \ie $(a,b)\in\Pi_k$ must be followed by $(b,\pi^{[l]}\circ b)$ from~$\Pi_l$, recall \eqref{tr_pairs_k}. Thus,
\ba
\label{tr_biv_1}
	p_{\tau;kl}(1)\ =\ 
    \sum_{(a,b)\in\Pi_k} P\Big(\pi_t=a, \pi_{t+1}=b,\ \pi_{t+2}=\pi^{[l]}\circ b\Big),
\ea
where $\pi^{[l]}\circ b$ is computed according to \eqref{tabOP}.

\bigskip
For the three types of process considered in Section~\ref{Marginal Distribution of Transcripts}, namely \iid, SRW, and GCT process, the OP process $(\pi_t)$ has a very short memory. More precisely, $(\pi_t)$ is 2-dependent for $(X_t)$ being \iid\ \citep{elsinger10,sousa22,weiss22}, and 1-dependent for both SRW and GCT \citep{sousa22,silbernagel26}. Here, $(\pi_t)$ is $q$-dependent with $q\in\bbn$ if~$\pi_t$ and~$\pi_{t+h}$ are independent if $h>q$. According to \eqref{tr_biv_h}, the $q$-dependence of $(\pi_t)$ implies the $(q+1)$-dependence of the transcript process $(\tau_t)$, \ie the bivariate probabilities only have to be computed for $h\in\{1,\ldots,q+1\}$.

\bigskip
Let us start with the bivariate transcript probabilities for lag $h=1$, where we have the simplified formula \eqref{tr_biv_1}. Using the short-hand notation $p_{\pi;uvw}=P\big(\pi_t=\pi^{[u]}, \pi_{t+1}=\pi^{[v]}, \pi_{t+2}=\pi^{[w]}\big)$ and considering impossible transitions between OPs, it generally holds that
\ba
\label{tr_biv_1_k1}
\begin{array}{@{}rl}
p_{\tau;1l}(1)\ =& 
\left\{\begin{array}{ll}
p_{\pi;111} + p_{\pi;666} & \text{if } l=1, \\
p_{\pi;112} & \text{if } l=2, \\
p_{\pi;664} & \text{if } l=3, \\
p_{\pi;663} & \text{if } l=4, \\
p_{\pi;115} & \text{if } l=5, \\
0 & \text{if } l=6; \\
\end{array}\right.\\[1ex]
\end{array}
\ea

\ba
\label{tr_biv_1_k2}
\begin{array}{@{}rl}
p_{\tau;2l}(1)\ =& 
\left\{\begin{array}{ll}
p_{\pi;566} & \text{if } l=1, \\
0 & \text{if } l=2, \\
p_{\pi;564} & \text{if } l=3, \\
p_{\pi;126} + p_{\pi;563} & \text{if } l=4, \\
p_{\pi;123} & \text{if } l=5, \\
p_{\pi;124} & \text{if } l=6; \\
\end{array}\right.\\[1ex]
\end{array}
\ea

\ba
\label{tr_biv_1_k3}
\begin{array}{@{}rl}
p_{\tau;3l}(1)\ =& 
\left\{\begin{array}{ll}
p_{\pi;311} & \text{if } l=1, \\
p_{\pi;312} & \text{if } l=2, \\
0 & \text{if } l=3, \\
p_{\pi;645} & \text{if } l=4, \\
p_{\pi;315} + p_{\pi;641} & \text{if } l=5, \\
p_{\pi;642} & \text{if } l=6; \\
\end{array}\right.\\[1ex]
\end{array}
\ea

\ba
\label{tr_biv_1_k4}
\begin{array}{@{}rl}
p_{\tau;4l}(1)\ =& 
\left\{\begin{array}{ll}
p_{\pi;266} & \text{if } l=1, \\
p_{\pi;456} & \text{if } l=2, \\
p_{\pi;264} + p_{\pi;631} & \text{if } l=3, \\
p_{\pi;263}+p_{\pi;326}+p_{\pi;632} & \text{if } l=4, \\
p_{\pi;323} + p_{\pi;454} & \text{if } l=5, \\
p_{\pi;324}+p_{\pi;453}+p_{\pi;635} & \text{if } l=6; \\
\end{array}\right.\\[1ex]
\end{array}
\ea

\ba
\label{tr_biv_1_k5}
\begin{array}{@{}rl}
p_{\tau;5l}(1)\ =& 
\left\{\begin{array}{ll}
p_{\pi;411} & \text{if } l=1, \\
p_{\pi;156}+p_{\pi;412} & \text{if } l=2, \\
p_{\pi;231} & \text{if } l=3, \\
p_{\pi;232}+p_{\pi;545} & \text{if } l=4, \\
p_{\pi;154} + p_{\pi;415}+p_{\pi;541} & \text{if } l=5, \\
p_{\pi;153}+p_{\pi;235}+p_{\pi;542} & \text{if } l=6; \\
\end{array}\right.\\[1ex]
\end{array}
\ea

\ba
\label{tr_biv_1_k6}
\begin{array}{@{}rl}
p_{\tau;6l}(1)\ =& 
\left\{\begin{array}{ll}
0 & \text{if } l=1, \\
p_{\pi;356} & \text{if } l=2, \\
p_{\pi;531} & \text{if } l=3, \\
p_{\pi;245}+p_{\pi;426}+p_{\pi;532} & \text{if } l=4, \\
p_{\pi;241} + p_{\pi;354}+p_{\pi;423} & \text{if } l=5, \\
p_{\pi;242}+p_{\pi;353}+p_{\pi;424}+p_{\pi;535} & \text{if } l=6. \\
\end{array}\right.\\[1ex]
\end{array}
\ea
Let us briefly explain the logic behind the proof with an example. To obtain the expression for $p_{\tau;1l}(1) = P\big(\tau_t=\pi^{[1]}, \tau_{t+1}=\pi^{[l]}\big)$ with $l=2$ in \eqref{tr_biv_1_k1}, first note that $p_{\tau;1} = P\big(\tau_t=\pi^{[1]}\big) = p_{\pi;11}(1) + p_{\pi;66}(1)$ according to \eqref{tr_marg2}.
Using that ``triplets'' starting with $a,b$ necessarily equal $a,b,\pi^{[l]}\circ b$, we just have to look into the $l$th column of \eqref{tabOP} in order to find that $p_{\tau;12}(1) = P\big(\tau_t=\pi^{[1]}, \tau_{t+1}=\pi^{[2]}\big) = p_{\pi;112} + p_{\pi;665}$. Finally, considering the impossible transitions $\{1,3,4\}\not\to\{3,4,6\}$ and $\{2,5,6\}\not\to\{1,2,5\}$, we recognize that $p_{\pi;665}=0$, so we end up with $p_{\tau;12}(1) = p_{\pi;112}$.

\bigskip
In order to explicitly calculate the lag-1 probabilities $\mP\!_\tau(1)$ from \eqref{tr_biv_1} for a specific type of OP-series, such as those investigated in Section~\ref{Marginal Distribution of Transcripts}, it should be noted that the OP-triplets in \eqref{tr_biv_1_k1}--\eqref{tr_biv_1_k6} correspond to a segment of length~5 of the original process $(X_t)$ and, thus, to OPs of length~5. So for a given pair $(\tau_t, \tau_{t+1}\big) = \big(\pi^{[k]}, \pi^{[l]}\big)$, we first need to determine the corresponding 5-OPs, and then we can compute the corresponding probability by summing up the probabilities of the 5-OPs. This strategy is analogous to the one in \citet[Table~II]{weiss22} and leads to Table~\ref{tabTrLag1} in Appendix~\ref{Tables}. Using this table, we can immediately compute the lag-1 probabilities $\mP\!_\tau(1)$ for the case where $(X_t)$ is \iid, because then, the 5-OPs are uniformly distributed with probability $1/5!=1/120$ each. So we simply have to count which combination of transcripts corresponds to how many 5-OPs.

\begin{prop}
\label{prop_tr_biv1_iid}
If $(X_t)$ is \iid, then the transcripts' bivariate lag-1 probabilities $\mP\!_\tau(1)=\big(p_{\tau;kl}(1)\big)_{k,l=1,\ldots,m!}$ are given by
$$
\mP\!_\tau(1)\ =\ 
\frac{1}{120}\,\left(\begin{array}{cccccc}
2 & 1 & 1 & 3 & 3 & 0 \\
1 & 0 & 1 & 6 & 1 & 1 \\
1 & 1 & 0 & 1 & 6 & 1 \\
3 & 1 & 6 & 14 & 4 & 7 \\
3 & 6 & 1 & 4 & 14 & 7 \\
0 & 1 & 1 & 7 & 7 & 4 \\
\end{array}\right).
$$
\end{prop}
Proposition~\ref{prop_tr_biv1_iid} can be used to conclude on the bivariate lag-1 probabilities of the Cayley and Kendall distances, where we introduce the short-hand notations $\mP\!_{\textup{C}}(h)=\big(p_{\textup{C};kl}(h)\big)_{k,l=0,1,2}$ with $p_{\textup{C};kl}(h) := P\big(d_{\textup{C},t}=k, d_{\textup{C},t+h}=l\big)$ and $\mP\!_{\textup{K}}(h)=\big(p_{\textup{K};kl}(h)\big)_{k,l=0,\ldots,3}$ with $p_{\textup{K};kl}(h) := P\big(d_{\textup{K},t}=k, d_{\textup{K},t+h}=l\big)$. 
Applying the transformation matrices from equations \eqref{trans_mat_C}--\eqref{trans_mat_K}, we immediately get the following results on the bivariate distributions as well as the respective autocorrelation function (acf), $\rho_{\textup{C}}(h) := Cor[d_{\textup{C},t}, d_{\textup{C},t+h}]$ and $\rho_{\textup{K}}(h) := Cor[d_{\textup{K},t}, d_{\textup{K},t+h}]$, respectively.

\begin{kor}
\label{kor_d_biv1_iid}
If $(X_t)$ is \iid, then the bivariate lag-1 probabilities of the Cayley and Kendall distances, respectively, are
$$
\mP\!_{\textup{C}}(1)\ =\mT_{\textup{C}}\, \mP\!_{\tau}(1)\, \mT_{\textup{C}}^{\top}
=\ \frac{1}{60}\,\left(
\begin{array}{ccc}
 1 & 1 & 3 \\
 1 & 5 & 14 \\
 3 & 14 & 18 \\
\end{array}
\right)\quad
\text{ with acf }
\rho_{\textup{C}}(1) = -\tfrac{2}{25}
$$
and
$$
\mP\!_{\textup{K}}(1)\ =\mT_{\textup{K}}\, \mP\!_{\tau}(1)\, \mT_{\textup{K}}^{\top}=\ \frac{1}{60}\,\left(
\begin{array}{cccc}
 1 & 1 & 3 & 0 \\
 1 & 1 & 7 & 1 \\
 3 & 7 & 18 & 7 \\
 0 & 1 & 7 & 2 \\
\end{array}
\right)\quad
\text{ with acf }
\rho_{\textup{K}}(1) = \tfrac{22}{115}.
$$
\end{kor}
We continue analogously with time lag $h=2$. From \eqref{tr_biv_h}, we obtain the lag-2 probabilities
\ba
\label{tr_biv_2}
	p_{\tau;kl}(2)\ =\ 
    \sum_{(a,b)\in\Pi_k} \sum_{(\alpha,\beta)\in\Pi_l} P\Big(\pi_t=a, \pi_{t+1}=b,\ \pi_{t+2}=\alpha, \pi_{t+3}=\beta\Big).
\ea
So we end up with quadruples of 3-OPs, which correspond to certain 6-OPs with respect to the original process $(X_t)$.
Like before, we determine all 6-OPs corresponding to the given pair $(\tau_t, \tau_{t+2}\big) = \big(\pi^{[k]}, \pi^{[l]}\big)$ of transcripts, see Table~\ref{tabTrLag2} in Appendix~\ref{Tables} for the result. Using this table, we compute the lag-2 probabilities $\mP\!_\tau(2)$ for the case where $(X_t)$ is \iid, because then, the 6-OPs are uniformly distributed with probability $1/6!=1/720$ each. So again, we simply have to count which combination of transcripts corresponds to how many 6-OPs.

\begin{prop}
\label{prop_tr_biv2_iid}
If $(X_t)$ is \iid, then the transcripts' bivariate lag-2 probabilities $\mP\!_\tau(2)=\big(p_{\tau;kl}(2)\big)_{k,l=1,\ldots,m!}$ are given by
$$
\mP\!_\tau(2)\ =\ 
\frac{1}{720}\,\left(\begin{array}{cccccc}
2 & 7 & 7 & 18 & 18 & 8 \\
7 & 0 & 14 & 25 & 6 & 8 \\
7 & 14 & 0 & 6 & 25 & 8 \\
18 & 6 & 25 & 77 & 48 & 36 \\
18 & 25 & 6 & 48 & 77 & 36 \\
8 & 8 & 8 & 36 & 36 & 24 \\
\end{array}\right).
$$
\end{prop}

Similarly as in Corollary~\ref{kor_d_biv1_iid}, we obtain:

\begin{kor}
\label{kor_d_biv2_iid}
If $(X_t)$ is \iid, then the bivariate lag-2 probabilities of the Cayley and Kendall distances, respectively, are
$$
\mP\!_{\textup{C}}(2)\ =\ \frac{1}{360}\,\left(
\begin{array}{ccc}
 1 & 11 & 18 \\
 11 & 42 & 67 \\
 18 & 67 & 125 \\
\end{array}
\right)\quad
\text{ with acf }
\rho_{\textup{C}}(2) = 0
$$
and
$$
\mP\!_{\textup{K}}(2)\ =\ \frac{1}{360}\,\left(
\begin{array}{cccc}
 1 & 7 & 18 & 4 \\
 7 & 14 & 31 & 8 \\
 18 & 31 & 125 & 36 \\
 4 & 8 & 36 & 12 \\
\end{array}
\right)\quad
\text{ with acf }
\rho_{\textup{K}}(2) = \tfrac{8}{115}.
$$
\end{kor}
Note that the Cayley distances are uncorrelated at lag~2, but they are not independent, because their bivariate distribution according to Corollary~\ref{kor_d_biv2_iid} differs from the product of the marginal distributions:
$$
\fp_{\textup{C}}\,\fp_{\textup{C}}^\top\ =\ \big(p_{\textup{C};k}\, p_{\textup{C};l}\big)_{k,l=0,1,2}\ =\ \frac{1}{144}\,\left(
\begin{array}{ccc}
 1 & 4 & 7 \\
 4 & 16 & 28 \\
 7 & 28 & 49 \\
\end{array}
\right).
$$
At this point, let us briefly recall that both the SRW and GCT process have a 1-dependent OP-series \citep{sousa22,silbernagel26} and, thus, a 2-dependent transcript series, see the paragraph below \eqref{tr_biv_1}. So Tables~\ref{tabTrLag1} and~\ref{tabTrLag2} will be useful for computing the full bivariate lag-$h$ distributions of their transcript series, which appears to be another interesting direction for future research. If the OP-series arises from an \iid\ process $(X_t)$, in turn, it is known to be 2-dependent such that its transcript series is even 3-dependent. Therefore, we also need to compute the lag-3 probabilities $\mP\!_\tau(3)$ in order to obtain the full bivariate lag-$h$ distributions of the transcript series corresponding to the \iid\ process $(X_t)$. This is done like before by counting how many 7-OPs belong to  given pair $(\tau_t, \tau_{t+3}\big) = \big(\pi^{[k]}, \pi^{[l]}\big)$ of transcripts (a full table is omitted this time), and by using that the 7-OPs are uniformly distributed with probability $1/7!=1/5040$ each.

\begin{prop}
\label{prop_tr_biv3_iid}
If $(X_t)$ is \iid, then the transcripts' bivariate lag-3 probabilities $\mP\!_\tau(3)=\big(p_{\tau;kl}(3)\big)_{k,l=1,\ldots,m!}$ are given by
$$
\mP\!_\tau(3)\ =\ 
\frac{1}{5040}\,\left(\begin{array}{cccccc}
42 & 35 & 35 & 119 & 119 & 70 \\
35 & 40 & 30 & 116 & 129 & 70 \\
35 & 30 & 40 & 129 & 116 & 70 \\
119 & 129 & 116 & 421 & 440 & 245 \\
119 & 116 & 129 & 440 & 421 & 245 \\
70 & 70 & 70 & 245 & 245 & 140 \\
\end{array}\right).
$$
\end{prop}
Due to the 3-dependence of the transcript series (if $(X_t)$ is \iid), it holds that $\mP\!_\tau(h) = \fp_{\tau}\,\fp_{\tau}^\top$ for lag $|h|\geq 4$.

\begin{kor}
\label{kor_d_biv3_iid}
If $(X_t)$ is \iid, then the bivariate lag-3 probabilities of the Cayley and Kendall distances, respectively, are
$$
\mP\!_{\textup{C}}(3)\ =\ \frac{1}{360}\,\left(
\begin{array}{ccc}
 3 & 10 & 17 \\
 10 & 40 & 70 \\
 17 & 70 & 123 \\
\end{array}
\right)
\text{ with acf }
\rho_{\textup{C}}(3) = \tfrac{1}{75}
$$
and
$$
\mP\!_{\textup{K}}(3)\ =\ \frac{1}{360}\,\left(
\begin{array}{cccc}
 3 & 5 & 17 & 5 \\
 5 & 10 & 35 & 10 \\
 17 & 35 & 123 & 35 \\
 5 & 10 & 35 & 10 \\
\end{array}
\right)
\text{ with acf }
\rho_{\textup{K}}(3) = \tfrac{1}{115}.
$$
\end{kor}
Again due to the 3-dependence, it holds that $\mP\!_{\textup{C}}(h) = \fp_{\textup{C}}\,\fp_{\textup{C}}^\top$ and $\mP\!_{\textup{K}}(h) = \fp_{\textup{K}}\,\fp_{\textup{K}}^\top$ for lag $|h|\geq 4$, as well as $\rho_{\textup{C}}(h) = \rho_{\textup{K}}(h) =0$.


\numberwithin{satz}{section}

\section{Asymptotics of Transcript Statistics}
\label{Asymptotics of Transcript Statistics}
In what follows, we are interested in the asymptotic distribution of statistics being based on a time series $\tau_1,\ldots,\tau_n$ of transcripts. More precisely, we estimate the true marginal distribution of transcripts, $\fp_{\tau}$, by the corresponding vector~$\widehat{\fp}_{\tau}$ of relative frequencies computed from $\tau_1,\ldots,\tau_n$, and the considered transcript statistics are functions of~$\widehat{\fp}_{\tau}$. As the starting point, we consider the asymptotics of the vector~$\widehat{\fp}_{\tau}$ itself. This vector can be expressed as the mean across binary indicator $\fZ_1, \ldots,\fZ_n$, which are defined by the ``one-hot'' encoding: $Z_{t,k}=\indfkt(\tau_t=\pi^{[k]})$ for $k=1,\ldots,m!$ and $t=1,\ldots,n$, where the indicator functions $\indfkt(A)$ equals~1 (0) if~$A$ is true (false). Thus, $\fp_{\tau} = E[\fZ_t]$ and $\widehat{\fp}_{\tau} = \overline{\fZ} := \frac{1}{n}\,\sum_{t=1}^n \fZ_t$. With an analogous argumentation as in \citet[Theorem~II.6]{silbernagel26}, we obtain the following central limit theorem, where~``$\0$'' denotes the vector of zeros and~``$\norm$'' abbreviates the (multivariate) normal distribution.

\begin{satz}
\label{satzCLTtranscripts}
Let the transcript series $(\tau_t)$ be stationary and $\alpha$-mixing with mixing coefficients $\alpha_i\geq 0$, $i\in\bbn_0$, satisfying $\sum_{i=0}^\infty \alpha_i<\infty$. Then, as $n\to\infty$,
$$
\sqrt{n}\,\big(\widehat{\fp}_{\tau}-\fp_{\tau}\big)\ \overset{d}{\to}\ \norm(\0,\fSigma_{\tau}),
$$
where the asymptotic covariance matrix~$\fSigma_{\tau}$ equals
$$
\fSigma_{\tau}\ =\ \textup{diag}(\fp_{\tau})-\fp_{\tau}\,\fp_{\tau}^\top\ +\ \sum_{h=1}^\infty \Big(\mP\!_\tau(h) + \mP\!_\tau(h)^\top - 2\,\fp_{\tau}\,\fp_{\tau}^\top\Big).
$$
\end{satz}
The $\alpha$-mixing condition is automatically satisfied if the transcript series $(\tau_t)$ is $q$-dependent with some $q\in\bbn$, \ie if $\mP\!_\tau(h) = \fp_{\tau}\,\fp_{\tau}^\top$ for all $|h|>q$. Then, the summation for $\fSigma$ even reduces to a finite sum for $h=1,\ldots,q$. Recall that SRW and GCT lead to a 2-dependent transcript series $(\tau_t)$, while the case of $(X_t)$ being \iid\ implies $(\tau_t)$ to be 3-dependent.
Hence, Theorem~\ref{satzCLTtranscripts} together with Propositions~\ref{prop_tr_marg_iid}, \ref{prop_tr_biv1_iid}, \ref{prop_tr_biv2_iid}, and~\ref{prop_tr_biv3_iid} immediately yields the following result.

\begin{kor}
\label{korCLTtranscripts_iid}
Let $(\tau_t)$ be the transcript series from an \iid\ process $(X_t)$, \ie with $\fp_{\tau} = 
\tfrac{1}{24}\,\big(2, 2, 2, 7, 7, 4\big)^\top$. Then, as $n\to\infty$,
$$
\sqrt{n}\,\big(\widehat{\fp}_{\tau}-\fp_{\tau}\big)\ \overset{d}{\to}\ \norm(\0,\fSigma_{\tau}),
$$
where the asymptotic covariance matrix~$\fSigma_{\tau}$ has rank~5 and equals
$$
\fSigma_{\tau}\ =\ \frac{1}{20\,160}\, \left(
\begin{array}{cccccc}
 1820 & 28 & 28 & -462 & -462 & -952 \\
 28 & 1020 & 380 & -406 & -406 & -616 \\
 28 & 380 & 1020 & -406 & -406 & -616 \\
 -462 & -406 & -406 & 6259 & -4453 & -532 \\
 -462 & -406 & -406 & -4453 & 6259 & -532 \\
 -952 & -616 & -616 & -532 & -532 & 3248 \\
\end{array}
\right).
$$
\end{kor}
The frequency vectors of the Cayley and Kendall distances follow similarly as in equation \eqref{p_C=T_C*p_tau} for the probability vectors:
$$
\widehat{\fp}_{\textup{C}}\ =\ \mT_{\textup{C}}\, \widehat{\fp}_{\tau},
\quad
\widehat{\fp}_{\textup{K}}\ =\ \mT_{\textup{K}}\, \widehat{\fp}_{\tau}.
$$
Hence, by the normal distribution's invariance with respect to affine transformations, it immediately follows that~$\widehat{\fp}_{\textup{C}}$ and~$\widehat{\fp}_{\textup{K}}$ are asymptotically normal as well, where the respective covariance matrices are equal to $\fSigma_{\textup{C}} = \mT_{\textup{C}}\, \fSigma_{\tau}\, \mT_{\textup{C}}^\top$ and $\fSigma_{\textup{K}} = \mT_{\textup{K}}\, \fSigma_{\tau}\, \mT_{\textup{K}}^\top$. Therefore, in the case of an \iid\ process $(X_t)$, Proposition~\ref{prop_tr_marg_iid} and Corollary~\ref{korCLTtranscripts_iid} yield the following result.

\begin{kor}
\label{korCLT_d_iid}
Let $(d_{\textup{C},t})$ and $(d_{\textup{K},t})$ be the series of Cayley and Kendall distances, respectively, from an \iid\ process $(X_t)$, \ie with $\fp_{\textup{C}} = \tfrac{1}{12}\,\big(1, 4, 7\big)^\top$ and $\fp_{\textup{K}} = \tfrac{1}{12}\,\big(1, 2, 7, 2\big)^\top$. Then, as $n\to\infty$,
$$
\sqrt{n}\,\big(\widehat{\fp}_{\textup{C}}-\fp_{\textup{C}}\big)\ \overset{d}{\to}\ \norm(\0,\fSigma_{\textup{C}})
\quad\text{and}\quad
\sqrt{n}\,\big(\widehat{\fp}_{\textup{K}}-\fp_{\textup{K}}\big)\ \overset{d}{\to}\ \norm(\0,\fSigma_{\textup{K}}),
$$
where the asymptotic covariance matrices~$\fSigma_{\textup{C}}$ and~$\fSigma_{\textup{K}}$ equal
$$
\fSigma_{\textup{C}}\ =\ \frac{1}{720}\, \left(
\begin{array}{ccc}
 65 & -32 & -33 \\
 -32 & 128 & -96 \\
 -33 & -96 & 129 \\
\end{array}
\right)
\quad\text{with rank } 2
$$
and
$$
\fSigma_{\textup{K}}\ =\ \frac{1}{720}\, \left(
\begin{array}{cccc}
 65 & 2 & -33 & -34 \\
 2 & 100 & -58 & -44 \\
 -33 & -58 & 129 & -38 \\
 -34 & -44 & -38 & 116 \\
\end{array}
\right)
\quad\text{with rank } 3.
$$
\end{kor}
Corollary~\ref{korCLT_d_iid} allows to deduce the asymptotics of the mean Cayley or Kendall distance, which is again computed as an affine transformation of the respective frequency:
$$
\begin{array}{rl}
\mu_{\textup{C}}\ =\ (0,1,2)\, \fp_{\textup{C}},
&
\overline{d_{\textup{C}}}\ =\ (0,1,2)\, \widehat{\fp}_{\textup{C}},
\\[1ex]
\mu_{\textup{K}}\ =\ (0,\ldots,3)\, \fp_{\textup{K}},
&
\overline{d_{\textup{K}}}\ =\ (0,\ldots,3)\, \widehat{\fp}_{\textup{K}}.
\end{array}
$$
So both mean distances~$\overline{d_{\textup{C}}}$ and~$\overline{d_{\textup{K}}}$ are asymptotically normally distributed with variance $(0,1,2)\, \fSigma_{\textup{C}}\, (0,1,2)^\top$ and $(0,\ldots,3)\, \fSigma_{\textup{K}}\, (0,\ldots,3)^\top$, respectively.

\begin{kor}
\label{korCLT_mean_d_iid}
Let $(d_{\textup{C},t})$ and $(d_{\textup{K},t})$ be the series of Cayley and Kendall distances, respectively, from an \iid\ process $(X_t)$, \ie with means $\mu_{\textup{C}} = \tfrac{3}{2}$ and $\mu_{\textup{K}} = \tfrac{11}{6}$. Then, as $n\to\infty$,
$$
\sqrt{n}\,\big(\,\overline{d_{\textup{C}}}-\mu_{\textup{C}}\big)\ \overset{d}{\to}\ \norm\big(0,\tfrac{13}{36}\big)
\quad\text{and}\quad
\sqrt{n}\,\big(\,\overline{d_{\textup{K}}}-\mu_{\textup{K}}\big)\ \overset{d}{\to}\ \norm\big(0,\tfrac{59}{60}\big).
$$
\end{kor}

\bigskip
A possible application of the above asymptotics is to develop tests for serial dependence, \ie tests of the null hypothesis that the process $(X_t)$ is \iid {} Recall that the asymptotics in Corollaries~\ref{korCLTtranscripts_iid}--\ref{korCLT_mean_d_iid} apply to any continuously distributed \iid\ process $(X_t)$, irrespective of its marginal distribution. Thus, the dependence tests derived from Corollaries~\ref{korCLTtranscripts_iid}--\ref{korCLT_mean_d_iid} are nonparametric (distribution-free), which is highly attractive for practice.

\smallskip
Two tests for serial dependence immediately get clear from Corollary~\ref{korCLT_mean_d_iid}. Under the \iid-null, the mean Cayley distance~$\overline{d_{\textup{C}}}$ and the mean Kendall distance~$\overline{d_{\textup{K}}}$ are asymptotically normally distributed as given in Corollary~\ref{korCLT_mean_d_iid}, so we can use these normal distributions for computing (two-sided) critical values or P-values of the $\overline{d_{\textup{C}}}$-test and $\overline{d_{\textup{K}}}$-test, respectively. Recall from parts~(a) and~(b) of Figure~\ref{figExpectedDistances} as well as from Example~\ref{bspAlternativeScenarios} that for some kinds of serial dependence, these mean distances are expected to deviate considerably from the null values $\mu_{\textup{C}} = \tfrac{3}{2}$ and $\mu_{\textup{K}} = \tfrac{11}{6}$, respectively, such that the $\overline{d_{\textup{C}}}$-test and $\overline{d_{\textup{K}}}$-test might turn out to be powerful for these alternatives.

\smallskip
The most common approach in the OP-community, however, is to consider entropy-like statistics \citep{bandt02,bandt19,weiss22}, which compare the full vector of marginal frequencies to the marginal distribution under the null (in the aforementioned works, the OP-frequencies $\widehat{\fp}_\pi$ are compared to $\fp_\pi^{(0)}=(\frac{1}{m!}, \ldots, \frac{1}{m!})^\top$, the marginal OP-distribution under the \iid-null). In the present research, we propose three options, namely to compare the

\begin{itemize}
    \item transcript frequencies $\widehat{\fp}_{\tau}$ to $\fp_{\tau}^{(0)} = 
\tfrac{1}{24}\,\big(2, 2, 2, 7, 7, 4\big)^\top$, see Corollary~\ref{korCLTtranscripts_iid};
    \item Cayley frequencies $\widehat{\fp}_{\textup{C}}$ to $\fp_{\textup{C}}^{(0)} = \tfrac{1}{12}\,\big(1, 4, 7\big)^\top$, see Corollary~\ref{korCLT_d_iid};
    \item Kendall frequencies $\widehat{\fp}_{\textup{K}}$ to $\fp_{\textup{K}}^{(0)} = \tfrac{1}{12}\,\big(1, 2, 7, 2\big)^\top$, see Corollary~\ref{korCLT_d_iid}.
\end{itemize}
Let~$\widehat{\fp}$ and~$\fp^{(0)}$ be such $k$-dimensional frequency and probability vectors, respectively. Then, the entropy (``deviance'') statistic $H\big(\widehat{\fp}\big)$, \ie the Kullback--Leibler divergence between~$\widehat{\fp}$ and~$\fp^{(0)}$, is given by
\ba
\label{entropy}
H\big(\widehat{\fp}\big)
\ =\ 
-\sum_{i=1}^k \hat{p}_i\,\big(\ln{\hat{p}_i} - \ln{p_i^{(0)}}\big)
\ \approx\ 
-\frac{1}{2}\,\sum_{i=1}^k \frac{\big(\hat{p}_i-p_i^{(0)}\big)^2}{p_i^{(0)}},
\ea
where the last expression is the second-order Taylor approximation of $H\big(\widehat{\fp}\big)$ around~$\fp^{(0)}$, see Equation~(6) in \citet{weiss22}. Except the factor~$-\frac{1}{2}$, this is the well-known Pearson statistic for goodness-of-fit,
\ba
\label{Delta2}
\Delta\big(\widehat{\fp}\big)
\ =\ 
\sum_{i=1}^k \frac{\big(\hat{p}_i-p_i^{(0)}\big)^2}{p_i^{(0)}}
\ =\ 
\big(\widehat{\fp}-\fp^{(0)}\big)^\top\, \textup{diag}\big(\fp^{(0)}\big)^{-1}\,\big(\widehat{\fp}-\fp^{(0)}\big),
\ea
which has also been considered for OP-frequencies in former research \citep{bandt19,weiss22}. Both statistics \eqref{entropy} and \eqref{Delta2} can be combined with the aforementioned three options of transcripts, Cayley and Kendall distances, leading to six further tests for serial dependence, the test statistics of which are denoted as $H_\tau\big(\widehat{\fp}_\tau\big)$, $H_{\textup{C}}\big(\widehat{\fp}_{\textup{C}}\big)$, $H_{\textup{K}}\big(\widehat{\fp}_{\textup{K}}\big)$ as well as $\Delta_\tau\big(\widehat{\fp}_\tau\big)$, $\Delta_{\textup{C}}\big(\widehat{\fp}_{\textup{C}}\big)$, $\Delta_{\textup{K}}\big(\widehat{\fp}_{\textup{K}}\big)$. Their asymptotic distributions under the \iid-null can be derived by applying Theorem~3.1 of \citet{tan77} together with Corollaries~\ref{korCLTtranscripts_iid} and~\ref{korCLT_d_iid} to $\Delta\big(\widehat{\fp}\big)\approx -2\, H\big(\widehat{\fp}\big)$. If the frequency vector is asymptotically normal of the form $\sqrt{n}\,\big(\widehat{\fp}-\fp^{(0)}\big) \overset{d}{\to} \norm(\0,\fSigma)$, then $n\,\Delta\big(\widehat{\fp}\big)$ and $-2n\, H\big(\widehat{\fp}\big)$ asymptotically follow the quadratic-form (QF) distribution $\sum_{i=1}^k \lambda_i\,\chi_{1;i}^2$, where the~$\chi_{1;i}^2$ are independent $\chi^2$-distributed random variables with one degree of freedom, and where the~$\lambda_i$ are the eigenvalues of the matrix $\fSigma\, \textup{diag}\big(\fp^{(0)}\big)^{-1}$. These eigenvalues can be explicitly calculated from the expressions provided by Corollaries~\ref{korCLTtranscripts_iid} and~\ref{korCLT_d_iid}, where we always have one eigenvalue equal to zero (note that the rank of the covariance matrices is always one less than the matrix dimension).

\begin{prop}
\label{propEntropies}
Let $(X_t)$ be an \iid\ process. Then,

\begin{itemize}
    \item[(a)] $-2n\,H_\tau\big(\widehat{\fp}_\tau\big)$ and $n\,\Delta_\tau\big(\widehat{\fp}_\tau\big)$ are asymptotically QF-distributed according to $\sum_{i=1}^5 \lambda_{\tau;i}\,\chi_{1;i}^2$, where
    $$
\lambda_{\tau;1} = \frac{1339}{735},
\quad
\lambda_{\tau;2} \approx 1.5468,
\quad
\lambda_{\tau;3} \approx 0.9260,
\quad
\lambda_{\tau;4} \approx 0.7177,
\quad
\lambda_{\tau;5} = \frac{8}{21};
    $$
    \item[(b)] $-2n\,H_{\textup{C}}\big(\widehat{\fp}_{\textup{C}}\big)$ and $n\,\Delta_{\textup{C}}\big(\widehat{\fp}_{\textup{C}}\big)$ are asymptotically QF-distributed according to $\sum_{i=1}^2 \lambda_{\textup{C};i}\,\chi_{1;i}^2$, where
    $$
\lambda_{\textup{C};1} = \frac{6}{5},
\quad
\lambda_{\textup{C};2} = \frac{76}{105};
    $$
    \item[(c)] $-2n\,H_{\textup{K}}\big(\widehat{\fp}_{\textup{K}}\big)$ and $n\,\Delta_{\textup{K}}\big(\widehat{\fp}_{\textup{K}}\big)$ are asymptotically QF-distributed according to $\sum_{i=1}^3 \lambda_{\textup{K};i}\,\chi_{1;i}^2$, where
    $$
\lambda_{\textup{K};1} \approx 1.5468,
\quad
\lambda_{\textup{K};2} \approx 0.9260,
\quad
\lambda_{\textup{K};3} \approx 0.7177.
    $$
\end{itemize}
\end{prop}
Note that (upper) critical values or P-values for the QF-distributions according to Proposition~\ref{propEntropies} can be computed numerically, \eg by using the R~package \href{https://cran.r-project.org/package=CompQuadForm}{\nolinkurl{CompQuadForm}} of \citet{duchesne10}.

\section{Empirical Investigations}
\label{Empirical Investigations}


\subsection{Simulation-based Performance Analyses}
\label{Simulation-based Performance Analyses}
In order to evaluate the performance of the nonparametric dependence tests developed in Section~\ref{Asymptotics of Transcript Statistics}, we did a simulation study, where time series of different lengths $T\in\{100, 250, 500, 1000\}$ from various types of process are simulated (with $10^4$ replications each). All tests are applied to all time series (where the level was chosen as 5\,\% for convenience), and the resulting rates of rejections are the empirical sizes (if the tests are applied under the \iid-null) or power values (under an alternative process), respectively. For evaluating the tests' sizes, we simulate \iid\ $\norm(0,1)$ time series---recall that the tests are nonparametric, so the actual marginal distribution does not matter. As alternative scenarios, we consider the processes discussed in Example~\ref{bspAlternativeScenarios}, which mostly correspond to processes in the simulation of \citet{weiss22}. Hence, we can directly compare the newly simulated power values to those presented by \citet{weiss22}, \ie the new tests are compared to various competitors. These are the parametric ACF-test on the one hand (as the most widely used test for serial dependence in practice), and various nonparametric tests being based on OP-frequencies (but not on transcripts as done here). To keep the comparison manageable, we focus on the best-performing OP-tests in \citet{weiss22}, namely the $H$-, $\Delta$-, $\upbeta$- and $\uptau$-tests developed there. Finally, we also complemented the alternative scenarios of Example~\ref{bspAlternativeScenarios} by a ``contamination experiment'' with additive outliers (AOs): for both AR$(1)$ processes, we randomly selected 10\,\% of the time-series observations and added either~$-10$ or $+10$ (with probability 50\,\% each). In this way, we can analyze the robustness of the tests with respect to extreme observations. All simulation results are summarized in Table~\ref{tabPerformance} in Appendix~\ref{Tables}.

\smallskip
The size values of the novel dependence tests (first block in Table~\ref{tabPerformance}) are reasonably close to the nominal 5\,\% level. In particular, there are at most modest exceedances of 5\,\%, while notable undersizing (\ie a conservative test) is only observed for the $\overline{d_{\textup{C}}}$-test if $T\leq 250$. But some competitors ($\upbeta$- and ACF-test) are even more conservative for $T=100$. The remaining rejection rates in Table~\ref{tabPerformance} are empirical power values for the alternative scenarios described before. Let us start with the uncontaminated AR$(1)$ process in blocks~2 and~4. It can be seen that among the novel tests, the $\overline{d_{\textup{K}}}$-test is most powerful throughout, but also the $H_{\textup{K}}$- and $\Delta_{\textup{K}}$-test lead to a similar power performance, and the $H_{\tau}$- and $\Delta_{\tau}$-test do only somewhat worse. In particular, these novel transcript-based tests lead to much better power values than the former OP-based tests of \citet{weiss22}, being close to the power values of the ACF, which is, in a sense, tailor-made for uncovering dependence in Gaussian AR processes. However, a major advantage of OP- and transcript-based tests gets clear if the AR time series are contaminated by outliers (blocks~3 and~5 in Table~\ref{tabPerformance}). Wile the power gets reduced for all tests, the  OP- and transcript-based tests are much more robust and still rather powerful, whereas the ACF is hardly powerful anymore. Also under AOs, $\overline{d_{\textup{K}}}$-test is most powerful among the novel tests, and now most powerful even among all tests considered. 

\smallskip
The remaining alternatives are nonlinear processes. For QMA$(1)$ processes, already \citet{weiss22} showed that some OP-based tests outperform the ACF, but the $H_{\tau}$- and $\Delta_{\tau}$-test are even more powerful and thus clearly perform best among all tests. By contrast, for the TEAR$(1)$ process, the $H_{\tau}$- and $\Delta_{\tau}$-test do not reach the power of the former OP-based tests. In fact, this is the only simulation scenario where tests of \citet{weiss22} are superior, which can be explained by a characteristic feature of the TEAR$(1)$ process: the generated time series are characterized by long lasting rises, \ie the OP~$\pi^{[1]}$ is very frequent. The former $H$-, $\Delta$-, and $\upbeta$-tests are very well suited for detecting such unusual OP-frequencies, whereas the transcript-based tests focus on ``differences'' between successive OPs rather than particular OPs themselves. Finally, for the ARCH$(1)$ process in the last block of Table~\ref{tabPerformance}, all tests are at most moderately powerful, which is not surprising as we applied the tests to the original time series, but not to the squared values (the latter would behave like an AR$(1)$ time series). Nevertheless, it is interesting to note that the tests based on the Cayley distance outperform all other OP-based and transcript-based tests, and that the $\overline{d_{\textup{C}}}$-test even outperforms the ACF for $T=1000$. The latter is somewhat surprising, because the $\overline{d_{\textup{C}}}$-test often performed rather poorly in the other alternative scenarios.

\smallskip
To sum up, the transcript-based $H_{\tau}$- and $\Delta_{\tau}$-tests show an appealing power for most of the considered alternatives. For AR$(1)$ processes, they are even surpassed by the $\overline{d_{\textup{K}}}$-test, which, however, performs rather poorly for the remaining alternatives. The $\overline{d_{\textup{C}}}$-test, by contrast, only excelled for the ARCH$(1)$ process. It is important to stress that with the exception of the TEAR$(1)$ process, the novel (nonparametric and robust) transcript-based tests are more powerful than the former OP-based tests, and often also outperform the ACF, such that they can be recommended as a valuable complement to existing tests for serial dependence. 


\subsection{Real-World Data Application}
\label{Real-World Data Application}
As an illustrative data example, we consider a time series about the monthly mean of air temperature (in ${}^\circ$C) at 2\,m above ground for the station ``Ham\-burg--Fuhlsb\"uttel'' in Germany in the period 1936--2024. The data source is the \href{https://cdc.dwd.de/portal/}{Climate Data Center (CDC)} offered by Deutscher Wetterdienst (German Weather Service), which was accessed on December~10, 2025. The originally 1068 mean temperatures are plotted in Figure~\ref{figTemp}\,(a) and show a strong seasonal pattern with period~12 (monthly data). Therefore, we applied the difference operator with lag~12 (yearly difference) to the data and obtained the time series $x_1,\ldots,x_T$ of length~$T=1056$ shown in Figure~\ref{figTemp}\,(b). These differenced data appear to be stationary without an apparent dependence structure, so it seems reasonable to test the null hypothesis that the data are serially independent, \ie that $x_1,\ldots,x_T$ arose from an \iid\ process.

\begin{figure}[t]
\centering\footnotesize
(a)\hspace{-3ex}\includegraphics[viewport=0 40 695 230, clip=, scale=0.5]{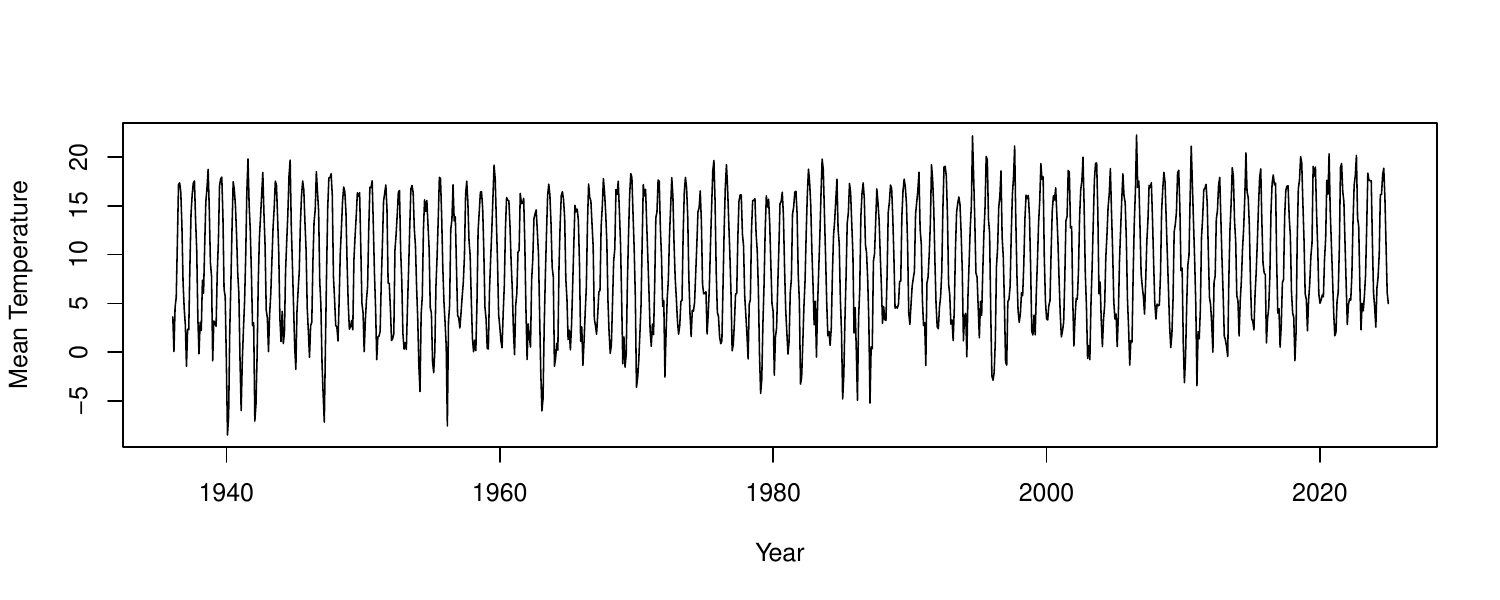}\hspace{-4ex}Year
\\[2ex]
(b)\hspace{-3ex}\includegraphics[viewport=0 40 695 230, clip=, scale=0.5]{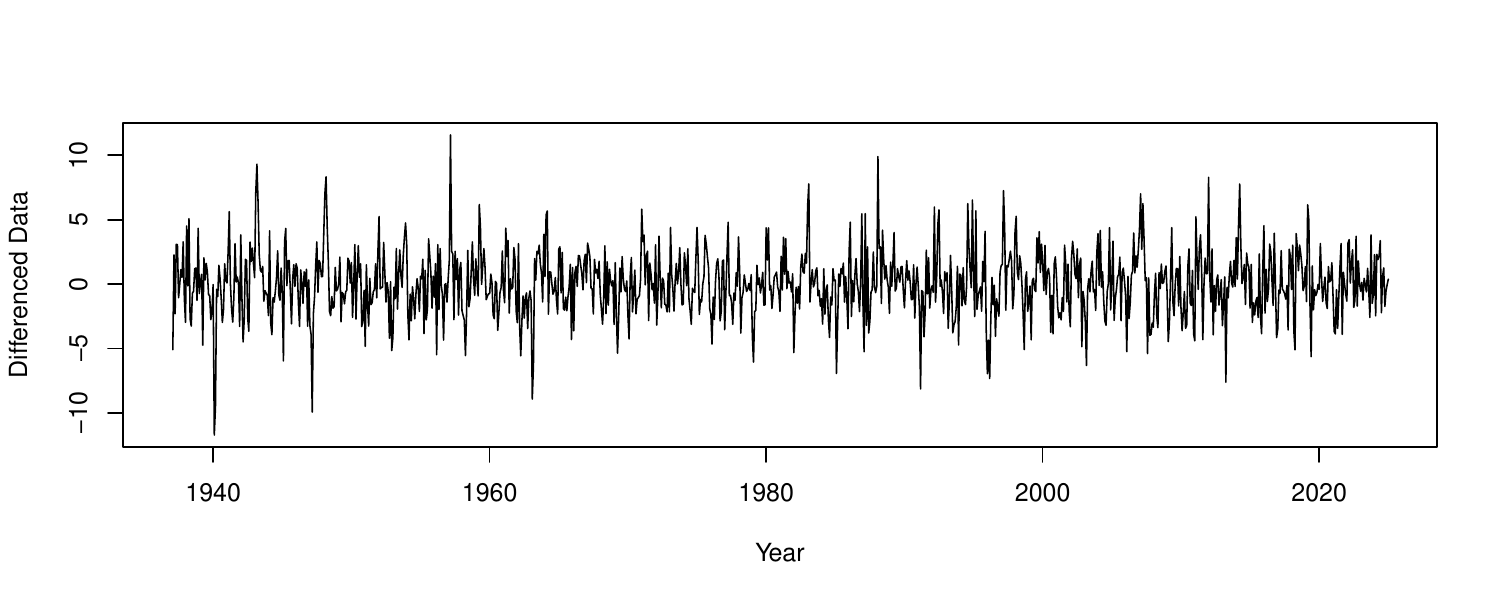}\hspace{-4ex}Year
\caption{Temperature data of Section~\ref{Real-World Data Application}: plot of raw data in~(a), and differenced data (lag~12) in (b).}
\label{figTemp}
\end{figure}

\smallskip
For this purpose, we apply our novel transcript-based tests developed in Section~\ref{Asymptotics of Transcript Statistics}. So we first compute the transcript series $\tau_1,\ldots,\tau_n$ as well as the distance series $d_{\textup{C},1}, \ldots, d_{\textup{C},n}$ and $d_{\textup{K},1}, \ldots, d_{\textup{K},n}$ of length $n=1053$ each, where distributional and moment properties are summarized in Table~\ref{tabDistMomTranscripts_data}. Comparing the empirically observed values to those of the \iid\ null, we recognize clear deviations throughout, indicating that the time series $(x_t)$ still exhibits serial dependence. In fact, comparing to Table~\ref{tabDistMomTranscripts} from Example~\ref{bspAlternativeScenarios}, it seems that we are concerned with positive dependence as, for example, the first transcript corresponding to the repeated occurrence of either OP~$\pi^{[1]}$ or~$\pi^{[6]}$ is more frequent than under an \iid\ process. This is also confirmed by the lag-1 ACF of $(x_t)$ being equal to $\approx 0.225$ as well as by the marginal frequencies of raw OPs, $\widehat{\fp}_\pi\approx (0.190, 0.150, 0.163, 0.150, 0.163, 0.184)^\top$, which deviates from a discrete uniform distribution and has increased frequencies for OPs~$\pi^{[1]}$ and~$\pi^{[6]}$. Note, however, that the transcript frequency $\hat{p}_{\tau,1}\approx 0.104$ provides another information than the marginal OP-frequencies $\hat{p}_{\pi,1}\approx 0.190$ and $\hat{p}_{\pi,6}\approx 0.184$, because the first transcript refers to the \emph{repeated} rather than a single occurrence of these OPs. In the context of our data example, the occurrence of transcript~1 implies that an increase (or decrease) of temperature differences across three months is continued to the fourth month, also recall the first column in Table~\ref{tabTrLag0}.

\begin{table}[t]
\centering\small
\caption{(Empirical) distributional and moment properties of (a) transcripts, (b) Cayley distance, and (c) Kendall distance for \iid\ process $(X_t)$ against temperature data $(x_t)$ from Section~\ref{Real-World Data Application}.}
\label{tabDistMomTranscripts_data}

\smallskip
(a)~\begin{tabular}{l|cccccc}
\toprule
Process & $p_{\tau; 1}$ & $p_{\tau; 2}$ & $p_{\tau; 3}$ & $p_{\tau; 4}$ & $p_{\tau; 5}$ & $p_{\tau; 6}$ \\
\midrule
\it \iid{} & \it 0.083 & \it 0.083 & \it 0.083 & \it 0.292 & \it 0.292 & \it 0.167 \\[1ex]
data & 0.104 & 0.098 & 0.091 & 0.273 & 0.297 & 0.137 \\
\bottomrule
\end{tabular}

\bigskip
(b)~\begin{tabular}{l|ccc@{\qquad}cc}
\toprule
Process & $p_{\textup{C}; 0}$ & $p_{\textup{C}; 1}$ & $p_{\textup{C}; 2}$ & $\mu_{\textup{C}}$ & $\sigma_{\textup{C}}^2$ \\
\midrule
\it \iid{} & \it 0.083 & \it 0.333 & \it 0.583 & \it 1.500 & \it 0.417 \\[1ex]
data & 0.104 & 0.326 & 0.570 & 1.465 & 0.458 \\
\bottomrule
\end{tabular}

\bigskip
(c)~\begin{tabular}{l|cccc@{\qquad}cc}
\toprule
Process & $p_{\textup{K}; 0}$ & $p_{\textup{K}; 1}$ & $p_{\textup{K}; 2}$ & $p_{\textup{K}; 3}$ & $\mu_{\textup{K}}$ & $\sigma_{\textup{K}}^2$ \\
\midrule
\it \iid{} & \it 0.083 & \it 0.167 & \it 0.583 & \it 0.167 & \it 1.833 & \it 0.639 \\[1ex]
data & 0.104 & 0.189 & 0.570 & 0.137 & 1.739 & 0.676 \\
\bottomrule
\end{tabular}
\end{table}

\smallskip
Next, we use the above transcript and distance frequencies to compute the test statistics and corresponding P-values according to Section~\ref{Asymptotics of Transcript Statistics}. The results are shown in Table~\ref{tabTests_data}, together with those of the same competing tests as in Section~\ref{Simulation-based Performance Analyses}. Note that for the sake of direct comparison to the respective asymptotic distribution, we show the rescaled versions of all test statistics, see Corollary~\ref{korCLT_mean_d_iid}, Proposition~\ref{propEntropies}, as well as the asymptotics in \citet{weiss22}. If comparing the P-values to the 5\,\%-level, we recognize from part~(a) that all transcript-based tests except~$\overline{d_{\textup{C}}}$ and~$H_{\textup{C}}$ lead to a rejection of the \iid-null. This outcome agrees with our power analyses in Section~\ref{Simulation-based Performance Analyses}, where the Cayley tests often performed rather poorly while transcripts and Kendall were often rather powerful. The test based on the lag-1 ACF according to Table~\ref{tabTests_data}\,(b) also recognizes significant dependence, whereas the OP-based tests lead to some ambiguity. The OP-based $H$- and $\Delta$-statistic have P-values only slightly below the level and thus lead to a very narrow rejection, which confirms our finding from Section~\ref{Simulation-based Performance Analyses}, where we noted a poorer power for the OP-based entropies than for the transcript-based ones. The $\upbeta$-test was powerful only for a few types of process, so it is not surprising that it does not lead to a rejection for the temperature data. The $\uptau$-test, in turn, was often quite powerful, which goes along with its rejection of the \iid-null in the present example. Altogether, the test decisions for the temperature data are well explainable with regard to our previous power analyses, and they confirm the appealing performance of our novel nonparametric tests being based on transcripts or Kendall distance.

\begin{table}[t]
\centering\footnotesize
\caption{Rescaled statistics and P-values of (a) transcript-based tests and (b) competing tests, see Section~\ref{Real-World Data Application}.}
\label{tabTests_data}

\smallskip
\begin{tabular}{l|c@{\ }c@{\ }c@{\ }c@{\ }c@{\ }c@{\ }c@{\ }c}
\toprule
(a) & $\sqrt{n}\,\big(\,\overline{d_{\textup{C}}}-\frac{3}{2}\big)$ & $\sqrt{n}\,\big(\,\overline{d_{\textup{K}}}-\frac{11}{6}\big)$ & $-2n\,H_\tau$ & $-2n\,H_{\textup{C}}$ & $-2n\,H_{\textup{K}}$ & $n\,\Delta_\tau$ & $n\,\Delta_{\textup{C}}$ & $n\,\Delta_{\textup{K}}$ \\
\midrule
Statistic & -1.125 & -3.066 & 15.965 & 5.734 & 14.592 & 16.153 & 6.155 & 14.773 \\
P-value & 0.061 & 0.002 & 0.020 & 0.053 & 0.005 & 0.019 & 0.043 & 0.005 \\
\bottomrule
\end{tabular}

\bigskip
\begin{tabular}{l|ccccc}
\toprule
(b) & $\frac{T-2}{3}\,(\ln{6}-H)$ & $(T-2)\,\Delta$ & $\sqrt{T-2}\,\hat{\upbeta}$ & $\sqrt{T-2}\,\hat{\uptau}$ & $\sqrt{T}\,\hat{\rho}(1)$ \\
\midrule
Statistic & 1.485 & 1.498 & 0.185 & 1.314 & 7.303 \\
P-value & 0.050 & 0.049 & 0.749 & 0.002 & 0.000 \\
\bottomrule
\end{tabular}
\end{table}


\section{Conclusions}
\label{Conclusions}
Since the introduction of OPs by \citet{bandt02}, the study of their applications to random
processes and dynamical systems is an active research topic, see \citet{amigo23} and the references therein. As a further step in this long-term endeavor, we focused here on the perhaps simplest way to exploit the group structure of the OPs, namely the concept of transcript (a group-theoretical difference of OPs) and the related Cayley and Kendall edit distances (norms of transcripts). What distinguishes this work from previous works on the properties and applications of transcripts \citep{monetti09,amigo12,monetti13} is the objective: nonparametric detection of serial
dependence via transcripts (including the Cayley and Kendall distances). Therefore, this paper continues the research of \citet{weiss22} and \citet{weiss22b} on the detection of serial dependence via OPs, this time by means of transcripts.

\smallskip
To achieve the aforementioned objective, in Section~\ref{Distributional Properties of Transcripts}, we first derived the marginal distributions of transcripts and distances when the underlying random process is \iid\ \eqref{prop_tr_marg_iid}, a symmetric random walk (Proposition~\ref{prop_tr_marg_srw}), or a generalized coin tossing process (Proposition~\ref{prop_tr_marg_gct}), obtaining results that showed the discrimination power of our transcript-based approach. Then, in Section~\ref{Bivariate Distributions of Transcripts}, we focused on \iid\ processes and derived the bivariate transcript distributions for lags $h\geq 1$. In Section~\ref{Asymptotics of Transcript Statistics}, we derived the asymptotic distribution of statistics based on the relative frequencies of transcripts and distances
for \iid\ processes (Corollaries~\ref{korCLTtranscripts_iid}--\ref{korCLT_mean_d_iid}). With these asymptotic
distributions at hand, the final step was to develop tests for the null hypothesis that the original process is \iid {} The statistics were denoted by $\overline{d_{\textup{C}}}$, $\overline{d_{\textup{K}}}$ (mean distances), $H_{\tau},H_{\textup{C}},H_{\textup{K}}$ (deviance or Kullback--Leibler entropy), and $\Delta _{\tau},\Delta _{\textup{C}},\Delta _{\textup{K}}$ (Pearson statistic for goodness of fit).

\smallskip
The power for serial dependence of the transcript-based statistics was analyzed with simulated data in Section~\ref{Simulation-based Performance Analyses} and real-world data in Section~\ref{Real-World Data Application}. The numerical processes comprised four linear processes (AR$(1)$~$\pm 0.5$ with or without AOs) and three nonlinear processes (QMA$(1)$, TEAR$(1)$, and ARCH$(1)$). Overall, the new statistics showed better rejection rates than former OP-based tests (including $H$, $\Delta$, and ACF). The real-world data consisted of time
series of monthly mean air temperatures. Again, the transcript-based
statistics and P-values generally showed a better performance than the competing OP-based statistics, which was explained in light of the results with simulated data. The results are summarized in Table~\ref{tabPerformance} of the Appendix.

\medskip
In conclusion, the novel transcript-based statistics proposed in this paper have proved to be powerful tools for the detection of serial dependence. In this regard, the results in Section~\ref{Empirical Investigations} make $H_{\tau}$, $H_{\textup{K}}$ (or $\Delta _{\tau}, \Delta _{\textup{K}}$, respectively) and $\overline{d_{\textup{K}}}$ particularly recommendable. More generally, our empirical investigations show that the algebraic structure of the OPs can be an advantage in certain applications, suggesting that further applications of transcript-based tools in time series analysis is a subject worth researching. Future work along lines similar to the present work will explore two extensions of the above framework: transcripts between two or more time series (\ie cross-transcripts) and transcripts between spatial ordinal patterns (\ie OPs defined for random fields). Both tools are relevant for studying coupled random processes and dynamical systems (since symbolization randomizes a deterministic dynamic). Furthermore, it would be relevant to obtain closed-form distributions and asymptotics for further types of process than those considered here.



\subsubsection*{Conflicts of Interest}
The authors declare no conflicts of interest.

\subsubsection*{Data Availability Statement}
Data and codes for Section~\ref{Real-World Data Application} are available in the online supplementary materials.



\clearpage


\appendix
\small
\numberwithin{equation}{section}
\numberwithin{table}{section}

\section{Tables}
\label{Tables}


\begin{table}[bh!]
\centering\footnotesize
\caption{Possible 5-OPs obtained if transcripts are observed with lag~1 (row followed by column). Here, $(i_1,\ldots,i_5)$ is abbreviated as ``$i_1\ldots i_5$''.}
\label{tabTrLag1}

\smallskip
\begin{tabular}{l|cccc@{\ }cc@{\ }cc}
\toprule
$\tau_t\ \setminus\ \tau_{t+1}$ & $\pi^{[1]}$ & $\pi^{[2]}$ & $\pi^{[3]}$ & \multicolumn{2}{c}{$\pi^{[4]}$} & \multicolumn{2}{c}{$\pi^{[5]}$} & $\pi^{[6]}$ \\
 \midrule
$\pi^{[1]}$ & 12345 & 12354 & 45321 & 43215 &  & 12534 &  & --- \\
 & 54321 &  &  & 43251 &  & 15234 &  &  \\
 &  &  &  & 43521 &  & 51234 &  &  \\[1ex]
$\pi^{[2]}$ & 54312 & --- & 45312 & 12543 &  & 12435 &  & 12453 \\
 &  &  &  & 15243 &  &  &  &  \\
 &  &  &  & 43125 &  &  &  &  \\
 &  &  &  & 43152 &  &  &  &  \\
 &  &  &  & 43512 &  &  &  &  \\
 &  &  &  & 51243 &  &  &  &  \\[1ex]
$\pi^{[3]}$ & 21345 & 21354 & --- & 53421 &  & 21534 &  & 35421 \\
 &  &  &  &  &  & 25134 &  &  \\
 &  &  &  &  &  & 34215 &  &  \\
 &  &  &  &  &  & 34251 &  &  \\
 &  &  &  &  &  & 34521 &  &  \\
 &  &  &  &  &  & 52134 &  &  \\[1ex]
$\pi^{[4]}$ & 15432 & 54231 & 14532 & 14325 & 32541 & 21435 &  & 21453 \\
 & 51432 &  & 32145 & 14352 & 35214 & 24135 &  & 24153 \\
 & 54132 &  & 32415 & 21543 & 35241 & 42531 &  & 24513 \\
 &  &  & 32451 & 25143 & 41325 & 45231 &  & 42315 \\
 &  &  & 41532 & 25413 & 41352 &  &  & 42351 \\
 &  &  & 45132 & 32154 & 52143 &  &  & 53214 \\
 &  &  &  & 32514 & 52413 &  &  & 53241 \\[1ex]
$\pi^{[5]}$ & 23145 & 15423 & 13245 & 13254 &  & 14253 & 34152 & 14235 \\
 & 23415 & 23154 &  & 13524 &  & 14523 & 34512 & 15324 \\
 & 23451 & 23514 &  & 53142 &  & 25314 & 41253 & 31542 \\
 &  & 23541 &  & 53412 &  & 25341 & 41523 & 35142 \\
 &  & 51423 &  &  &  & 31425 & 45123 & 35412 \\
 &  & 54123 &  &  &  & 31452 & 52314 & 41235 \\
 &  &  &  &  &  & 34125 & 52341 & 51324 \\[1ex]
$\pi^{[6]}$ & --- & 54213 & 31245 & 15342 &  & 13425 &  & 13542 \\
 &  &  &  & 25431 &  & 13452 &  & 24531 \\
 &  &  &  & 31254 &  & 24315 &  & 42135 \\
 &  &  &  & 31524 &  & 24351 &  & 53124 \\
 &  &  &  & 35124 &  & 42153 &  &  \\
 &  &  &  & 51342 &  & 42513 &  &  \\
 &  &  &  & 52431 &  & 45213 &  &  \\
 \bottomrule
 \end{tabular}
\end{table}

\clearpage


\begin{sidewaystable}[th!]
\centering
\caption{Possible 6-OPs obtained if transcripts are observed with lag~2 (row followed by column). Here, $(i_1,\ldots,i_6)$ is abbreviated as ``$i_1\ldots i_6$''.}
\label{tabTrLag2}

\smallskip
\resizebox{.95\linewidth}{!}{
\begin{tabular}{l|c@{\ }cc@{\ }c@{\ }cc@{\ }c@{\ }cc@{\ }c@{\ }c@{\ }c@{\ }c@{\ }c@{\ }c@{\ }cc@{\ }c@{\ }c@{\ }c@{\ }c@{\ }c@{\ }c@{\ }cc@{\ }c@{\ }c@{\ }c}
\toprule
$\tau_t\ \setminus\ \tau_{t+1}$ & \multicolumn{2}{c}{$\pi^{[1]}$} & \multicolumn{3}{c}{$\pi^{[2]}$} & \multicolumn{3}{c}{$\pi^{[3]}$} & \multicolumn{8}{c}{$\pi^{[4]}$} & \multicolumn{8}{c}{$\pi^{[5]}$} & \multicolumn{4}{c}{$\pi^{[6]}$} \\
 \midrule
$\pi^{[1]}$ & 123456 &  & 123465 &  &  & 432156 &  &  & 123654 & 463251 &  &  &  &  &  &  & 123546 & 453261 &  &  &  &  &  &  & 123564 &  &  &  \\
 & 654321 &  & 126534 &  &  & 432516 &  &  & 126354 & 463521 &  &  &  &  &  &  & 123645 & 453621 &  &  &  &  &  &  & 125346 &  &  &  \\
 &  &  & 162534 &  &  & 432561 &  &  & 162354 & 543216 &  &  &  &  &  &  & 125364 & 456321 &  &  &  &  &  &  & 152346 &  &  &  \\
 &  &  & 165234 &  &  & 435216 &  &  & 432165 & 543261 &  &  &  &  &  &  & 125634 & 512364 &  &  &  &  &  &  & 465321 &  &  &  \\
 &  &  & 612534 &  &  & 435261 &  &  & 432615 & 543621 &  &  &  &  &  &  & 126345 & 512634 &  &  &  &  &  &  & 512346 &  &  &  \\
 &  &  & 615234 &  &  & 435621 &  &  & 432651 & 546321 &  &  &  &  &  &  & 152364 & 516234 &  &  &  &  &  &  & 643215 &  &  &  \\
 &  &  & 651234 &  &  & 564321 &  &  & 436215 & 612354 &  &  &  &  &  &  & 152634 & 561234 &  &  &  &  &  &  & 643251 &  &  &  \\
 &  &  &  &  &  &  &  &  & 436251 & 645321 &  &  &  &  &  &  & 156234 & 612345 &  &  &  &  &  &  & 643521 &  &  &  \\
 &  &  &  &  &  &  &  &  & 436521 &  &  &  &  &  &  &  & 162345 &  &  &  &  &  &  &  &  &  &  &  \\
 &  &  &  &  &  &  &  &  & 463215 &  &  &  &  &  &  &  & 453216 &  &  &  &  &  &  &  &  &  &  &  \\[1ex]
$\pi^{[2]}$ & 126543 &  & --- &  &  & 124356 & 512643 &  & 124365 & 431652 & 543162 &  &  &  &  &  & 124536 &  &  &  &  &  &  &  & 124653 &  &  &  \\
 & 162543 &  &  &  &  & 125643 & 516243 &  & 124635 & 436125 & 543612 &  &  &  &  &  & 124563 &  &  &  &  &  &  &  & 126435 &  &  &  \\
 & 165243 &  &  &  &  & 152643 & 561243 &  & 125436 & 436152 & 546312 &  &  &  &  &  & 453126 &  &  &  &  &  &  &  & 162435 &  &  &  \\
 & 612543 &  &  &  &  & 156243 & 564312 &  & 125463 & 436512 & 612453 &  &  &  &  &  & 453162 &  &  &  &  &  &  &  & 465312 &  &  &  \\
 & 615243 &  &  &  &  & 431256 &  &  & 126453 & 463125 & 645312 &  &  &  &  &  & 453612 &  &  &  &  &  &  &  & 612435 &  &  &  \\
 & 651243 &  &  &  &  & 431526 &  &  & 152436 & 463152 &  &  &  &  &  &  & 456312 &  &  &  &  &  &  &  & 643125 &  &  &  \\
 & 654312 &  &  &  &  & 431562 &  &  & 152463 & 463512 &  &  &  &  &  &  &  &  &  &  &  &  &  &  & 643152 &  &  &  \\
 &  &  &  &  &  & 435126 &  &  & 162453 & 512436 &  &  &  &  &  &  &  &  &  &  &  &  &  &  & 643512 &  &  &  \\
 &  &  &  &  &  & 435162 &  &  & 431265 & 512463 &  &  &  &  &  &  &  &  &  &  &  &  &  &  &  &  &  &  \\
 &  &  &  &  &  & 435612 &  &  & 431625 & 543126 &  &  &  &  &  &  &  &  &  &  &  &  &  &  &  &  &  &  \\[1ex]
$\pi^{[3]}$ & 213456 &  & 213465 & 621534 &  & --- &  &  & 213654 &  &  &  &  &  &  &  & 213546 & 354261 & 563421 &  &  &  &  &  & 213564 &  &  &  \\
 & 342156 &  & 216534 & 625134 &  &  &  &  & 216354 &  &  &  &  &  &  &  & 213645 & 354621 & 621345 &  &  &  &  &  & 215346 &  &  &  \\
 & 342516 &  & 261534 & 652134 &  &  &  &  & 261354 &  &  &  &  &  &  &  & 215364 & 364215 & 634215 &  &  &  &  &  & 251346 &  &  &  \\
 & 342561 &  & 265134 & 653421 &  &  &  &  & 365421 &  &  &  &  &  &  &  & 215634 & 364251 & 634251 &  &  &  &  &  & 356421 &  &  &  \\
 & 345216 &  & 342165 &  &  &  &  &  & 621354 &  &  &  &  &  &  &  & 216345 & 364521 & 634521 &  &  &  &  &  & 521346 &  &  &  \\
 & 345261 &  & 342615 &  &  &  &  &  & 635421 &  &  &  &  &  &  &  & 251364 & 521364 &  &  &  &  &  &  & 534216 &  &  &  \\
 & 345621 &  & 342651 &  &  &  &  &  &  &  &  &  &  &  &  &  & 251634 & 521634 &  &  &  &  &  &  & 534261 &  &  &  \\
 &  &  & 346215 &  &  &  &  &  &  &  &  &  &  &  &  &  & 256134 & 526134 &  &  &  &  &  &  & 534621 &  &  &  \\
 &  &  & 346251 &  &  &  &  &  &  &  &  &  &  &  &  &  & 261345 & 536421 &  &  &  &  &  &  &  &  &  &  \\
 &  &  & 346521 &  &  &  &  &  &  &  &  &  &  &  &  &  & 354216 & 562134 &  &  &  &  &  &  &  &  &  &  \\[1ex]
$\pi^{[4]}$ & 165432 & 621543 & 321465 &  &  & 143256 & 413256 & 561432 & 143265 & 214635 & 254163 & 362514 & 423615 & 514632 & 542361 & 632541 & 145326 & 321546 & 352461 & 451326 & 532641 &  &  &  & 146532 & 321564 & 426531 & 621435 \\
 & 216543 & 625143 & 324165 &  &  & 143526 & 413526 & 562143 & 143625 & 215436 & 254613 & 362541 & 423651 & 521436 & 542631 & 635214 & 145362 & 321645 & 362145 & 451362 & 536214 &  &  &  & 164325 & 325164 & 461532 & 624135 \\
 & 261543 & 625413 & 324615 &  &  & 143562 & 413562 & 562413 & 143652 & 215463 & 261453 & 365214 & 426315 & 521463 & 546132 & 635241 & 145632 & 325146 & 362415 & 451632 & 536241 &  &  &  & 164352 & 325614 & 462531 & 641325 \\
 & 265143 & 651432 & 324651 &  &  & 156432 & 423156 & 564132 & 146325 & 216453 & 264153 & 365241 & 426351 & 524136 & 546231 & 641532 & 214536 & 325416 & 362451 & 452316 & 563214 &  &  &  & 214653 & 325641 & 465132 & 641352 \\
 & 265413 & 652143 & 653214 &  &  & 214356 & 423516 & 564231 & 146352 & 241365 & 264513 & 413265 & 461325 & 524163 & 614532 & 642531 & 214563 & 325461 & 415326 & 452361 & 563241 &  &  &  & 216435 & 352164 & 465231 & 642315 \\
 & 321456 & 652413 & 653241 &  &  & 215643 & 423561 &  & 154326 & 241635 & 321654 & 413625 & 461352 & 524613 & 621453 & 645132 & 241536 & 326145 & 415362 & 452631 & 632145 &  &  &  & 241653 & 352614 & 532146 & 642351 \\
 & 324156 & 654132 &  &  &  & 241356 & 516432 &  & 154362 & 246135 & 326154 & 413652 & 462315 & 541326 & 624153 & 645231 & 241563 & 326415 & 415632 & 456132 & 632415 &  &  &  & 246153 & 352641 & 532416 &  \\
 & 324516 & 654231 &  &  &  & 251643 & 521643 &  & 154632 & 251436 & 326514 & 416325 & 462351 & 541362 & 624513 &  & 245136 & 326451 & 425316 & 456231 & 632451 &  &  &  & 246513 & 356214 & 532461 &  \\
 & 324561 &  &  &  &  & 256143 & 526143 &  & 164532 & 251463 & 326541 & 416352 & 514326 & 541632 & 632154 &  & 245163 & 352146 & 425361 & 532164 &  &  &  &  & 261435 & 356241 & 614325 &  \\
 & 615432 &  &  &  &  & 256413 & 526413 &  & 214365 & 254136 & 362154 & 423165 & 514362 & 542316 & 632514 &  & 245613 & 352416 & 425631 & 532614 &  &  &  &  & 264135 & 416532 & 614352 &  \\[1ex]
$\pi^{[5]}$ & 132456 & 341526 & 132465 & 314652 & 651324 & 142356 &  &  & 132654 & 163524 & 316542 & 514623 & 623514 &  &  &  & 132546 & 153624 & 236451 & 315462 & 364125 & 451623 & 531642 & 623415 & 132564 & 235614 & 416253 & 531462 \\
 & 165423 & 341562 & 165324 & 341265 & 652314 & 156423 &  &  & 136254 & 164253 & 361542 & 541236 & 623541 &  &  &  & 132645 & 156324 & 253164 & 316425 & 364152 & 456123 & 536142 & 623451 & 135264 & 235641 & 416523 & 534126 \\
 & 231456 & 345126 & 231465 & 341625 & 652341 & 412356 &  &  & 136524 & 164523 & 365142 & 541263 & 631542 &  &  &  & 135246 & 163245 & 253614 & 316452 & 364512 & 513264 & 536412 & 631425 & 135624 & 253146 & 461253 & 534162 \\
 & 234156 & 345162 & 234165 & 341652 & 653142 & 516423 &  &  & 142365 & 231654 & 365412 & 541623 & 635142 &  &  &  & 136245 & 231546 & 253641 & 351426 & 412536 & 513624 & 561324 & 631452 & 142653 & 253416 & 461523 & 534612 \\
 & 234516 & 345612 & 234615 & 346125 & 653412 & 561423 &  &  & 142635 & 236154 & 412365 & 546123 & 635412 &  &  &  & 142536 & 231645 & 256314 & 351462 & 412563 & 516324 & 562314 & 634125 & 146253 & 253461 & 465123 & 614235 \\
 & 234561 & 615423 & 234651 & 346152 &  & 564123 &  &  & 146235 & 236514 & 412635 & 613254 & 641253 &  &  &  & 142563 & 235146 & 256341 & 354126 & 415236 & 523164 & 562341 & 634152 & 146523 & 315642 & 513246 & 641235 \\
 & 314256 & 651423 & 265314 & 346512 &  &  &  &  & 154236 & 236541 & 416235 & 613524 & 641523 &  &  &  & 145236 & 235416 & 263145 & 354162 & 415263 & 523614 & 563142 & 634512 & 153246 & 351642 & 523146 &  \\
 & 314526 & 654123 & 265341 & 615324 &  &  &  &  & 154263 & 263154 & 461235 & 614253 & 645123 &  &  &  & 145263 & 235461 & 263415 & 354612 & 415623 & 523641 & 563412 &  & 164235 & 356142 & 523416 &  \\
 & 314562 &  & 314265 & 625314 &  &  &  &  & 154623 & 263514 & 514236 & 614523 &  &  &  &  & 145623 & 236145 & 263451 & 361425 & 451236 & 526314 & 613245 &  & 231564 & 356412 & 523461 &  \\
 & 341256 &  & 314625 & 625341 &  &  &  &  & 163254 & 263541 & 514263 & 623154 &  &  &  &  & 153264 & 236415 & 315426 & 361452 & 451263 & 526341 & 623145 &  & 235164 & 412653 & 531426 &  \\[1ex]
$\pi^{[6]}$ & 134256 &  & 134265 &  &  & 243156 &  &  & 136542 & 264531 & 462135 & 631254 &  &  &  &  & 135426 & 245631 & 425163 & 536124 &  &  &  &  & 135642 & 351624 & 531246 &  \\
 & 134526 &  & 134625 &  &  & 243516 &  &  & 163542 & 312654 & 524316 & 631524 &  &  &  &  & 135462 & 312546 & 425613 & 561342 &  &  &  &  & 153426 & 356124 & 624315 &  \\
 & 134562 &  & 134652 &  &  & 243561 &  &  & 243165 & 316254 & 524361 & 635124 &  &  &  &  & 136425 & 312645 & 452136 & 563124 &  &  &  &  & 153462 & 421653 & 624351 &  \\
 & 265431 &  & 165342 &  &  & 256431 &  &  & 243615 & 316524 & 524631 & 642153 &  &  &  &  & 136452 & 315246 & 452163 & 613425 &  &  &  &  & 246531 & 426153 & 642135 &  \\
 & 312456 &  & 312465 &  &  & 421356 &  &  & 243651 & 361254 & 542136 & 642513 &  &  &  &  & 153642 & 316245 & 452613 & 613452 &  &  &  &  & 264315 & 426513 &  &  \\
 & 625431 &  & 615342 &  &  & 526431 &  &  & 246315 & 361524 & 542163 & 645213 &  &  &  &  & 156342 & 351246 & 456213 & 631245 &  &  &  &  & 264351 & 462153 &  &  \\
 & 652431 &  & 651342 &  &  & 562431 &  &  & 246351 & 365124 & 542613 &  &  &  &  &  & 163425 & 361245 & 513642 &  &  &  &  &  & 312564 & 462513 &  &  \\
 & 654213 &  & 653124 &  &  & 564213 &  &  & 254316 & 421365 & 546213 &  &  &  &  &  & 163452 & 421536 & 516342 &  &  &  &  &  & 315264 & 465213 &  &  \\
 &  &  &  &  &  &  &  &  & 254361 & 421635 & 613542 &  &  &  &  &  & 245316 & 421563 & 531264 &  &  &  &  &  & 315624 & 513426 &  &  \\
 &  &  &  &  &  &  &  &  & 254631 & 426135 & 624531 &  &  &  &  &  & 245361 & 425136 & 531624 &  &  &  &  &  & 351264 & 513462 &  &  \\
 \bottomrule
 \end{tabular}}
\end{sidewaystable}

\clearpage


\begin{sidewaystable}[th!]
\centering
\caption{Simulated rejection rates (size or power) from $10^4$ replications for different dependence tests and time series lengths~$T$, see Section~\ref{Simulation-based Performance Analyses}. Italic numbers are taken from \citet{weiss22}. Highest power among novel tests in bold font.}
\label{tabPerformance}

\smallskip
\resizebox{.9\linewidth}{!}{
\begin{tabular}{lr|cccc@{\qquad}c@{\qquad\quad}cc@{\qquad}ccc@{\qquad}ccc}
\toprule
Process & $T$ & $H$ & $\Delta$ & $\upbeta$ & $\uptau$ & ACF & $\overline{d_{\textup{C}}}$ & $\overline{d_{\textup{K}}}$ & $H_\tau$ & $H_{\textup{C}}$ & $H_{\textup{K}}$ & $\Delta_\tau$ & $\Delta_{\textup{C}}$ & $\Delta_{\textup{K}}$ \\
 \midrule
\iid{} & 100 & \it 0.051 & \it 0.047 & \it 0.042 & \it 0.055 & \it 0.039 & 0.043 & 0.051 & 0.052 & 0.052 & 0.052 & 0.049 & 0.050 & 0.048 \\
 & 250 & \it 0.054 & \it 0.052 & \it 0.059 & \it 0.047 & \it 0.047 & 0.043 & 0.048 & 0.046 & 0.052 & 0.047 & 0.044 & 0.050 & 0.047 \\
 & 500 & \it 0.051 & \it 0.050 & \it 0.050 & \it 0.048 & \it 0.044 & 0.051 & 0.051 & 0.051 & 0.050 & 0.049 & 0.049 & 0.050 & 0.049 \\
 & 1000 & \it 0.049 & \it 0.049 & \it 0.051 & \it 0.051 & \it 0.046 & 0.050 & 0.049 & 0.050 & 0.051 & 0.049 & 0.050 & 0.051 & 0.049 \\
 \midrule
AR$(1)$, & 100 & \it 0.250 & \it 0.256 & \it 0.051 & \it 0.556 & \it 0.998 & 0.576 & \bf 0.684 & 0.505 & 0.482 & 0.627 & 0.551 & 0.545 & 0.662 \\
$+0.5$ & 250 & \it 0.614 & \it 0.630 & \it 0.065 & \it 0.882 & \it 1.000 & 0.912 & \bf 0.973 & 0.919 & 0.865 & 0.962 & 0.932 & 0.885 & 0.967 \\
 & 500 & \it 0.929 & \it 0.934 & \it 0.057 & \it 0.992 & \it 1.000 & 0.996 & \bf 1.000 & 0.999 & 0.991 & \bf 1.000 & 0.999 & 0.992 & \bf 1.000 \\
 & 1000 & \it 0.999 & \it 0.999 & \it 0.057 & \it 1.000 & \it 1.000 & \bf 1.000 & \bf 1.000 & \bf 1.000 & \bf 1.000 & \bf 1.000 & \bf 1.000 & \bf 1.000 & \bf 1.000 \\
 \midrule
AR$(1)$, & 100 & 0.145 & 0.148 & 0.056 & 0.348 & 0.076 & 0.405 & \bf 0.422 & 0.289 & 0.318 & 0.379 & 0.329 & 0.375 & 0.418 \\
$+0.5$, & 250 & 0.334 & 0.346 & 0.067 & 0.647 & 0.143 & 0.757 & \bf 0.796 & 0.650 & 0.682 & 0.761 & 0.686 & 0.711 & 0.783 \\
AOs & 500 & 0.663 & 0.673 & 0.058 & 0.898 & 0.261 & 0.958 & 0.970 & 0.933 & 0.922 & 0.968 & 0.943 & 0.930 & \bf 0.972 \\
 & 1000 & 0.950 & 0.952 & 0.065 & 0.996 & 0.447 & 0.999 & \bf 1.000 & 0.999 & 0.998 & \bf 1.000 & 0.999 & 0.998 & \bf 1.000 \\
 \midrule
AR$(1)$, & 100 & 0.388 & 0.339 & 0.019 & 0.689 & 0.999 & 0.057 & \bf 0.841 & 0.740 & 0.324 & 0.822 & 0.710 & 0.217 & 0.802 \\
$-0.5$ & 250 & 0.855 & 0.835 & 0.023 & 0.976 & 1.000 & 0.161 & \bf 0.997 & 0.990 & 0.704 & 0.996 & 0.989 & 0.638 & 0.996 \\
 & 500 & 0.995 & 0.995 & 0.022 & 1.000 & 1.000 & 0.383 & \bf 1.000 & \bf 1.000 & 0.950 & \bf 1.000 & \bf 1.000 & 0.936 & \bf 1.000 \\
 & 1000 & 1.000 & 1.000 & 0.027 & 1.000 & 1.000 & 0.699 & \bf 1.000 & \bf 1.000 & \bf 1.000 & \bf 1.000 & \bf 1.000 & \bf 1.000 & \bf 1.000 \\
 \midrule
AR$(1)$, & 100 & 0.206 & 0.172 & 0.021 & 0.430 & 0.105 & 0.050 & \bf 0.549 & 0.419 & 0.221 & 0.504 & 0.388 & 0.138 & 0.477 \\
$-0.5$, & 250 & 0.518 & 0.488 & 0.027 & 0.802 & 0.172 & 0.122 & \bf 0.902 & 0.817 & 0.512 & 0.889 & 0.806 & 0.446 & 0.885 \\
AOs & 500 & 0.871 & 0.862 & 0.022 & 0.980 & 0.276 & 0.273 & \bf 0.996 & 0.986 & 0.803 & 0.994 & 0.986 & 0.769 & 0.993 \\
 & 1000 & 0.996 & 0.996 & 0.025 & 1.000 & 0.477 & 0.533 & \bf 1.000 & \bf 1.000 & 0.983 & \bf 1.000 & \bf 1.000 & 0.980 & \bf 1.000 \\
 \midrule
QMA$(1)$ & 100 & \it 0.157 & \it 0.156 & \it 0.168 & \it 0.099 & \it 0.139 & 0.088 & 0.085 & \bf 0.223 & 0.088 & 0.085 & 0.222 & 0.114 & 0.101 \\
 & 250 & \it 0.330 & \it 0.335 & \it 0.369 & \it 0.150 & \it 0.159 & 0.149 & 0.132 & \bf 0.502 & 0.157 & 0.132 & 0.500 & 0.178 & 0.151 \\
 & 500 & \it 0.587 & \it 0.595 & \it 0.590 & \it 0.234 & \it 0.161 & 0.274 & 0.197 & \bf 0.822 & 0.252 & 0.210 & 0.821 & 0.280 & 0.232 \\
 & 1000 & \it 0.888 & \it 0.891 & \it 0.872 & \it 0.431 & \it 0.161 & 0.456 & 0.332 & \bf 0.989 & 0.461 & 0.379 & \bf 0.989 & 0.478 & 0.404 \\
 \midrule
TEAR$(1)$ & 100 & \it 0.246 & \it 0.244 & \it 0.275 & \it 0.096 & \it 0.258 & 0.077 & 0.093 & \bf 0.172 & 0.084 & 0.086 & 0.171 & 0.106 & 0.096 \\
 & 250 & \it 0.534 & \it 0.539 & \it 0.611 & \it 0.137 & \it 0.539 & 0.115 & 0.160 & \bf 0.375 & 0.145 & 0.139 & \bf 0.375 & 0.165 & 0.151 \\
 & 500 & \it 0.824 & \it 0.829 & \it 0.868 & \it 0.219 & \it 0.792 & 0.194 & 0.263 & 0.667 & 0.234 & 0.234 & \bf 0.668 & 0.257 & 0.253 \\
 & 1000 & \it 0.984 & \it 0.985 & \it 0.992 & \it 0.378 & \it 0.969 & 0.324 & 0.461 & \bf 0.938 & 0.417 & 0.426 & 0.937 & 0.434 & 0.440 \\
 \midrule
ARCH$(1)$ & 100 & 0.061 & 0.060 & 0.064 & 0.082 & 0.274 & 0.110 & 0.077 & 0.071 & 0.096 & 0.080 & 0.078 & \bf 0.118 & 0.100 \\
 & 250 & 0.076 & 0.077 & 0.081 & 0.108 & 0.356 & \bf 0.193 & 0.101 & 0.092 & 0.153 & 0.116 & 0.102 & 0.172 & 0.136 \\
 & 500 & 0.094 & 0.096 & 0.067 & 0.141 & 0.422 & \bf 0.361 & 0.135 & 0.143 & 0.255 & 0.192 & 0.159 & 0.276 & 0.217 \\
 & 1000 & 0.153 & 0.156 & 0.077 & 0.251 & 0.493 & \bf 0.594 & 0.198 & 0.264 & 0.459 & 0.357 & 0.285 & 0.476 & 0.382 \\
 \bottomrule
 \end{tabular}}
\end{sidewaystable}

\end{document}